\documentclass[11pt,a4paper]{article}
\pdfoutput=1
\usepackage{amssymb,amsmath,amsthm,titling}
\usepackage[english]{babel}
\usepackage{tikz,eucal}
\usepackage[top=2.5cm,right=3cm,left=3cm,bottom=2.5cm]{geometry}

\usepackage{color}
\definecolor{verde}{cmyk}{.83,.21,1,.08}

\usepackage[pdfstartview=FitH,colorlinks=true,linkcolor=black,anchorcolor=black,citecolor=black,urlcolor=black
]{hyperref}

\allowdisplaybreaks[4]

\numberwithin{equation}{section}

\usepackage[titles]{tocloft}

\cftsetpnumwidth{9pt}
\newtheorem{prop}{Proposition}[section]

\newtheorem{lemma}[prop]{Lemma}
\newtheorem{cor}[prop]{Corollary}
\newtheorem{df}[prop]{Definition}
\newtheorem{rem}[prop]{Remark}

\theoremstyle{definition}
\newtheorem{ex}[prop]{Example}

\newcommand{\F}{\mathbb{H}}

\newcommand{\be}{\begin{equation}}
\newcommand{\ee}{\end{equation}}

\newcommand{\pa}{{\mathcal{P}(\A)}}
\newcommand{\sa}{{\mathcal{S}(\A)}}

\newcommand{\A}{\mathcal{A}}
\newcommand{\bS}{\mathbb{S}}
\newcommand{\T}{\mathbb{T}}
\newcommand{\B}{\mathcal{B}}
\newcommand{\K}{\mathcal{K}}

\newcommand{\OO}{\mathcal{O}}

\newcommand{\HH}{\mathcal{H}}
\newcommand{\N}{\mathbb{N}}
\newcommand{\NN}{\mathcal{N}}
\newcommand{\Z}{\mathbb{Z}}
\newcommand{\R}{\mathbb{R}}
\newcommand{\C}{\mathbb{C}}
\newcommand{\CP}{\mathbb{C}\mathrm{P}}
\newcommand{\I}{\mathbb{I}}
\newcommand{\de}{\mathrm{d}}
\newcommand{\inner}[1]{\left<#1\right>}
\newcommand{\mat}[1]{\bigg(\!\begin{array}{cc}#1\end{array}\!\bigg)}

\newcommand{\tr}{\mathrm{Tr}}

\newcommand{\mc}{\mathcal}
\newcommand{\mf}{\mathfrak}
\newcommand{\id}{\textup{\textsf{id}}}
\newcommand{\ketbra}[2]{\left|\smash[t]{#1}\right>\!\left<\smash[t]{#2}\right|}
\newcommand{\ket}[1]{\left|\smash[t]{#1}\right>}
\newcommand{\qA}{\mathbb{A}}
\newcommand{\qH}{\mathbb{H}}
\newcommand{\qD}{\mathbb{D}}

\begin{document}

\setlength{\droptitle}{-3pc}
\pretitle{\begin{flushright}\small
ICCUB-13-069
\end{flushright}\vspace*{2pc}%
\begin{center}\LARGE}
\posttitle{\par\end{center}}

\title{Spectral geometry with a cut-off:\\[3pt] topological and metric aspects \\[20pt]}

\author{\hspace*{10pt}Francesco D'Andrea$^{1,3}$, Fedele Lizzi$\hspace{1pt}^{2,3,4}$ and Pierre Martinetti$\hspace{1pt}^{2,3}$ \\[12pt]
{\footnotesize $^1$ Dipartimento di Matematica e Applicazioni, Universit\`a di Napoli {\sl Federico II}. } \\[3pt]
{\footnotesize $^2$ Dipartimento di Fisica, Universit\`a di Napoli {\sl Federico II}.}\\
{\footnotesize $^3$ I.N.F.N. -- Sezione di Napoli.} \\[3pt]
{\footnotesize $^4$ Departament de Estructura i Constituents de la Mat\`eria.} \\
{\footnotesize Institut de Ci\'encies del Cosmos, Universitat de Barcelona.}}

\date{}

\maketitle

\begin{abstract}
Inspired by regularization in quantum field theory, we study topological and metric properties of spaces in which a cut-off is introduced. We work in the framework
of noncommutative geometry, and focus on Connes distance associated to a spectral triple $(\A, \HH, D)$. A high momentum (short distance) cut-off is implemented by the action of a projection $P$ on the Dirac operator $D$ and/or on the algebra $\A$. This action induces two new distances. We individuate conditions making them equivalent to the original distance. We also study the Gromov-Hausdorff limit of the set of truncated states, first for compact quantum metric spaces in the sense of Rieffel, then for arbitrary spectral triples. To this aim, we introduce a notion of ``state with finite moment of order $1$'' for noncommutative algebras. We then focus on the commutative case, and show that the cut-off induces a minimal length between points, which is infinite if $P$ has finite rank.
When $P$ is a spectral projection of $D$, we work out an approximation of points by non-pure states that are at finite distance from each other. On the circle, such approximations are given by Fej\'er probability distributions.  Finally we apply the results to Moyal plane and the fuzzy sphere, obtained as Berezin quantization of the plane and the sphere respectively. 
\end{abstract}

\pagebreak

\vspace*{3cm}

\begin{center}
\begin{minipage}{14cm}
\tableofcontents
\end{minipage}
\end{center}

\pagebreak

%%% ======================================================================

\section{Introduction}
We study the  topological and metric aspects of spaces in which a cut-off is implemented. The physical motivation is the presence of divergent quantities in quantum field theory, and the techniques used to tame them. Specifically, the calculations in quantum field theory are usually done in the Fourier space and there are divergences due to the high values of the momentum. To obtain finite quantities, one introduces a large scale (usually called a cut-off) which represents the maximum attainable momentum. 
In this paper we are interested in the geometrical consequence of the cut-off: by momentum/position duality, cutting away high momenta means cutting away short distances. This means that the usual tools of differential geometry do not apply anymore, but the setting is ideal for the methods of noncommutative geometry~\cite{Con94,GVF01,Lan02,CM08}. The latter provides a translation of Riemannian geometry in completely algebraic terms, using a $*$-algebra $\A$ represented on a Hilbert space $\cal H$ (which capture the topological aspects), and a not necessarily bounded generalized Dirac operator $D$ (which captures the metric aspects). These elements form a \emph{spectral triple} and are at the basis of the construction. These ingredients are naturally present in any quantum field theory: the algebra is the one of complex-valued functions on spacetime $M$, the Hilbert space is the one representing the matter fields of the theory, and the generalized Dirac operator contains the relevant physical information of the system.

A cut-off $\Lambda>0$ is naturally implemented through the action of a suitably chosen projection. For instance when $M$ is closed, a natural way to define a regularized partition function is to decompose the infinite dimensional $\HH=\bigoplus_{\lambda}V_\lambda$ as the (completed) direct sum of eigenspaces of the Dirac operator $D$. Let $P_\Lambda$ denote the projection on the direct sum of eigenspaces with eigenvalue $\lambda$ not greater (in absolute value) than the cut-off,  $|\lambda|\leq \Lambda$. Since $D$ has compact resolvent, the projection $P_\Lambda$ and the operator
\begin{equation}\label{eq:28}
D_\Lambda := P_\Lambda D P_\Lambda \;
\end{equation} are of finite rank. The determinant of $D_\Lambda$ (thought of as an operator on $\HH_\Lambda:=P_\Lambda\HH$)
is well defined and gives a regularized partition function. This procedure is called \emph{finite mode regularization} \cite{AndrianovBonora1, Fujikawa, AKL11}. Although it is very much in the spirit of Noncommutative Geometry, it was in fact originally developed~before~it.

As recalled in \S\ref{basics}, for $M$ a Riemannian spin manifold the Dirac operator $D$ induces a distance $d_{\A, D}$ on the state space of the algebra $\A = C_0(M)$. This distance coincides with the geodesic distance if the states are pure (Dirac deltas), and with the Wasserstein distance of order $1$ of transport theory -
with cost the geodesic distance - if the states are given by arbitrary probability distributions on $M$ (see e.g.~\cite{DM09}). In this paper we give an answer to the question: ``how the metric properties of the state space changes upon the replacement of $D$ with $D_\Lambda$?''
In particular, we investigate whether the regularized theory is an ``approximation'' of the original theory in some precise mathematical sense.

\medskip

Although in finite mode regularization $P_\Lambda$ is an eigenprojection of $D$, i.e.~$[D,P_\Lambda]~=~0$, we work under the general hypothesis that $P_\Lambda$ is any projection on $\HH$, non-necessarily commuting with $D$. Furthermore, our starting point is any (not necessarily commutative) unital spectral triple $(\A,\HH,D)$. Except for sections \ref{sec5} and \ref{sectionBerezin}, we do not assume this is the canonical spectral triple of a closed Riemannian spin manifold. 
\medskip

The main results of that paper are the following:

\begin{list}{$-$}{\itemsep=0pt \itemindent=2em \leftmargin=0em}
\item\emph{Equivalence of the topologies induced by truncated distances:} we introduce two new distances - $d_{\A, D_\Lambda}$ and  $d_{\OO_\Lambda, D_\Lambda}$- corresponding to the truncation of the Dirac operator $D$ only, and to the truncation of both the Dirac operator and the algebra. The main result is Prop.~\ref{propequiv}, in which we work out some conditions on $D$ and $P_\Lambda$ under which these two distances are equivalent to the initial one $d_{\A, D}$. With weaker conditions, we also obtain in Prop.~\ref{propineq} some inequalities between the three distances.

\item\emph{Approximation of states and Gromov-Hausdorff convergence}: for compact quantum metric spaces in the sense of Rieffel, we show in Prop.~\ref{prop:3.12} that
any normal state of the initial algebra $\A$ can be approximated in the metric topology of $d_{\A, D}$ by a sequence of truncated states. We also show in Prop.~\ref{prop:4.3} that the convergence holds true not only for individual state, but also for metric spaces, in the Gromov-Hausdorff sense. Similar results are obtained in Prop.~\ref{propconvmetric} and Cor.~\ref{cormetrictop} for spectral triples that are not quantum metric spaces. There, normal states are substituted by a noncommutative generalization of states with finite moment of order $1$, a notion which is introduced in Def.~\ref{defmoment}. The difference between the weak$^*$ and the metric topologies is illustrated on a simple example: the lattice $\Z$. It is shown in Prop.~\ref{conlattice} that the closure of the inductive limit of truncated normal states of $C_0(\Z)$ is the space of normal states for the weak$^*$ topology, and the space of states with finite moment of order $1$ for the metric topology.

\item\emph{Approximation of points}: we show in Prop.~\ref{prop:1} that on a \emph{commutative} spectral geometry with a cut-off, the distance $d_{\A, D_\Lambda}$ between points (i.e.~pure states) is never smaller than the cut-off, and is even infinite when $P_\Lambda$ has finite rank (Prop.~\ref{propfiniterank}). Coming back to the original physical motivation, namely for $P_\Lambda$ an eigenprojection of the Dirac operator, we show in Prop.~\ref{prop:6.10} how to approximate points by non-pure states that remain at finite distance from one another. Specifically on the real line, we work out an approximation of points by non-pure states such that both distances $d_{\A,D_\Lambda},d_{\OO_\Lambda,D_\Lambda}$ are finite, and the latter actually coincides with $d_{\A, D}$ between points (Prop.~\ref{prop:realline}). On the circle, we approximate points by the non-pure state given by the evaluation of the Fej\'er transform of $f$. We show in Prop.~\ref{lemma:dgeo} that the distances are always smaller than the geodesic one and converge to it as the rank of the Fej\'er transform goes to infinity. We also provide a tight lower bound. %Alternative approximations of points on the circle are investigated in \S\ref{sec:6.2.1}.

\item\emph{Wasserstein distance and Berezin quantization}: in Prop.~\ref{prop:MoyalBerezin} we recall how Moyal plane can be seen as the complex plane with a cut-off, taking for $P_\Lambda$ the projection  on holomorphic functions (Berezin-Toeplitz quantization). We obtain in Prop.~\ref{cor:coherent} a new proof that the distance between coherent states of Moyal plane is the Euclidean distance. 
Similar techniques are applied to the sphere and allows to obtain new results on the metric properties of the fuzzy sphere (Prop.~\ref{ineqfuzzysphere} and \ref{prop:fuzzysphere}).
\end{list}

\section{Preliminaries}\label{basics}
We recall some basics on the metric aspect of Connes noncommutative geometry, Rieffel theory of compact quantum metric spaces, and Hausdorff convergence.

\subsection{Metric aspect of noncommutative geometry}\label{section2.1}
A spectral triple $(\A,\HH,D)$ is the datum of a separable Hilbert space $\HH$, a $*$-subalgebra $\A\subset\B(\HH)$ and a self-adjoint operator $D$ on $\HH$ such that, for all $a\in\A$, $[D,a]\in\B(\HH)$ and $a(D+i)^{-1}\in\K(\HH)$. We say that $(\A,\HH,D)$ is \emph{unital} if $\A$ is a unital algebra. The latter condition is then equivalent to $D$ having compact resolvent.

Although one can work with real Hilbert spaces and algebras (as for the spectral triple of the Standard Model \cite{AC2M2}), we assume that $\HH$ and $\A$ are complex. With some additional assumptions, from any commutative unital spectral triple one reconstructs an underlying Riemaniann manifold $M$ \cite{connesreconstruct}. That is the reason why a spectral triple over a noncommutative algebra is viewed as the noncommutative analogue of a manifold.

\smallskip

A state of a $C^*$-algebra $A$ is a positive linear functional $\varphi:A\to\C$ with norm $1$. The set of states, denoted $\sa$, is convex, with extreme points the
\emph{pure states}. It is often convenient to work with a dense subalgebra $\A$ of $A$ (like $C^\infty_0(M)\subset C_0(M)$). In that case, by a state of $\A$ we mean a state of its $C^*$-completion $A$. 

An important class of states is given by normal states. They are usually defined for von Neumann algebras as completely additive states \cite[Def.~7.1.1]{Kadison1983}. We
use the following equivalent definition \cite[Thm.~2.4.21]{BR96}, generalized to $C^*$-algebras \cite{Robinson:1970fk}\cite[Def.~2.4.25]{BR96}): a state $\varphi$ of $\A\subset\B(\HH)$ is \emph{normal} if there exists a positive trace-class operator $R$ on $\HH$ with $\tr(R)=1$, called \emph{density matrix}, such that
\begin{equation}\label{eq:normalstate}
\varphi(a)=\tr(Ra) \qquad\forall\;a\in\A \;.
\end{equation}
We denote by $\mathcal{N}(\A)$ the set of all normal states of $\A$. Notice that the map \eqref{eq:normalstate} from density matrices to ${\cal N}(\A)$ is surjective but not always injective: as soon as $\A$ is not the whole of $\B(\HH)$, there may be different density matrices associated to the same state.

\smallskip

Given $(\A, \HH, D)$ with $\A$ a (pre) $C^*$-algebra, $\,\mc{S}(\A)$ is an extended metric space\footnote{An extended metric space is a pair $(X,d)$ with $X$ a
set and $d: X \times X \to [0,\infty]$ a symmetric map satisfying the triangle inequality and such that $d(x,y) = 0$ iff $x = y$. It differs from an ordinary metric only in that the value $+\infty$ is allowed.} with distance
\begin{equation}\label{eq:spectraldistance}
d_{\A,D}(\varphi,\varphi') :=
\sup_{a\in\A^{\mathrm{sa}}}\big\{\,\varphi(a)-\varphi'(a)\,:\,
L_D(a)\leq 1 \,\big\}
\end{equation}
for all $\varphi,\varphi'\in \mc{S}(\A)$, where $L_D$ denotes the
seminorm defined on $\A$ by the operator $D$,
\begin{equation}\label{eq:5}
L_D(a):= \|[D,a]\|.
\end{equation}
We refer to it as the \emph{spectral distance}. Although in the original definition \cite{Con89} the supremum is over all $a\in\A$ obeying the side condition, it was noted in~\cite{Iochum:2001fv} that the supremum can be equivalently searched on the set $\A^{\mathrm{sa}}$ of self-adjoint elements of $\A$.

When $\A=C_0^\infty(M)$ for $M$ a Riemannian (spin) manifold and $D$ is a Dirac type operator, the spectral distance \eqref{eq:spectraldistance} on pure states coincides with the geodesic distance of the Riemannian metric.\footnote{Any point $x\in M$ is recovered as the pure state ``evaluation at $x$'', $\delta_x(f) := f(x)$, and any pure state of $C_0^\infty(M)$ comes from a point.} On arbitrary states, if $M$ is complete, it coincides with the Wasserstein distance of optimal transport theory (see e.g.~\cite{DM09}).

\subsection{Compact quantum metric spaces}
An \emph{order unit space} \cite{Kadison1983} is a real partially ordered vector space $\OO$ with a distinguished element $e$, called the \emph{order unit}, such that:
i) $\forall\;a\in \OO\;\exists\;r\in\R$ such that $a\leq re$; ii) if $a\leq re\;\forall\;r>0$, then $a\leq 0$. A norm on $\OO$ is given by
\begin{equation}\label{eq:58}
  \|a\| := \inf \big\{ r>0 : -re \leq a \leq re\big\} \,.
\end{equation}
A state on $\OO$ is a bounded linear map $\varphi:\OO\to\R$ with norm $1$, that is \cite[Thm.~4.3.2]{Kadison1983} $\varphi(e)=1$.  States are automatically positive. The collection $\mc{S}(\OO)$ of all states of $\OO$ is a compact topological space with respect to the weak$^*$ topology.

\smallskip

Any real vector subspace $\OO$ of $\B(\HH)^{\text{sa}}$ containing the identity $1$ is an order unit space for the partial ordering of operators, with order unit $e=1$. Actually any order unit space comes in this way \cite{Rie04}, so it makes sense to talk about normal states for order~unit~spaces.

\smallskip

A seminorm $L$ on $\OO$ defines on $\mc{S}(\OO)$ an extended metric
\begin{equation}\label{eq:orderunitd}
\rho_L(\varphi,\varphi') := \sup_{a\in\OO}\big\{\,\varphi(a)-\varphi'(a)\,:\,L(a) \leq 1 \,\big\} \;.
\end{equation}
Given a \emph{unital} spectral triple $(\A,\HH ,D)$, taking $\OO=\A^{\mathrm{sa}}$ and $L=L_D$, one recovers the spectral distance \eqref{eq:spectraldistance}. 
The seminorm $L$ is called \emph{Lipschitz} \cite{Rieffel:1999ec} if $L(a)=0$ implies $a\in\R e$. This is a necessary (but not sufficient) condition in order for $\rho_L$ to be finite.

\smallskip

If $X$ is a compact metric space and $\OO=C(X,\R)$, then
\begin{equation}
L(f)=\sup_{x\neq y}|f(x)-f(y)|/d(x,y)
\label{eq:21}
\end{equation}
is a Lipschitz seminorm and the associated metric $\rho_L$ induces on $\mc{S}(\OO)$ the weak$^*$ topology. This motivates the definition of a \emph{compact quantum metric space} \cite{rieffel2003,Rieffel:1999ec} as a pair $(\OO,L)$ such that $L$ is Lipschitz and $\rho_L$ induces on $\mc{S}(\OO)$ the weak$^*$ topology. These two conditions
guarantee that $\rho_L$ is finite on $S(\OO)$ \cite[Thm.~2.1]{rieffel2003}.

Locally compact quantum metric spaces have been recently introduced in \cite{Latremoliere:2012fk}.
An approach based on von Neumann algebras is in \cite{KW10}.

\subsection{Hausdorff and Gromov-Hausdorff distance}\label{sec:GH}
Let $X,Y$ be subsets of a metric space $(M,d)$, $d(x,Y):=\inf\nolimits_{y\in Y}d(x,y)$ the distance between $x\in X$ and the set $Y$, and $d(X,Y):=\sup\nolimits_{x\in X}d(x,Y)$ the largest possible distance between a point of $X$ and the set $Y$. The Hausdorff distance between $X$ and $Y$ is (see e.g.~\cite{BBI01,Gro07}):
\begin{equation}
d_H(X,Y):=\max \big\{d(X,Y),d(Y,X) \big\} \;.
\end{equation}
It is a semi-metric on the set of subsets of $M$ (meaning that distinct subsets of $M$ can be at zero distance), as $d_H(X,Y)=0$ iff $X$ and $Y$ have the same closure.
It becomes an extended metric if we consider only closed subsets of $M$ \cite[Prop.~7.3.3]{BBI01}. In particular, $d_H$ is an extended metric on the collection of compact subsets of $M$, and a metric if $M$ has finite diameter.\footnote{What we call ``extended metric'' is simply called a ``metric'' in \cite{BBI01}.}

A net of subsets $X_k$ of $M$ has Hausdorff limit $X\subset M$ if $\lim d_H(X_k,X)=0$.  This limit may not be unique, but it becomes unique if we require $X$ to be closed.

\smallskip

The Gromov-Hausdorff distance $d_{GH}(X,Y)$ between two complete metric spaces $(X,d_X)$ and $(Y,d_Y)$ is the infimum of the Hausdorff distance $d_H(f(X),g(Y))$ over all the isometric embeddings $f:X\to M$ and $g:Y\to M$ into a metric space $M$ \cite{Gro07}. It is generalized to compact quantum metric spaces by Rieffel in \cite{Rie04}. For subsets of a metric space $(M,d)$, the Hausdorff convergence implies the Gromov-Hausdorff convergence.

\section{Truncations}\label{truncations}
The regularization procedure motivated by quantum field theory, consisting in cutting off the spectrum of $D$, is implemented by the action of a finite-rank projection $P_\Lambda\in\B(\HH)$. Substituting in \eqref{eq:spectraldistance} the Dirac operator with $D_\Lambda$ as in \eqref{eq:28} modifies the spectral distance. In this section we study the relation between the distances associated to $D$ and~$D_\Lambda$. 
\subsection{Regularization of the geometry}
Since $D_\Lambda$ is bounded (it has finite rank) and compact operators form a two-sided ideal in $\B(\HH)$, if $b=a(D_\Lambda +i)^{-1}$ is compact then $a=b(D_\Lambda +i)$ is compact too. Thus $a(D_\Lambda +i)^{-1}$ cannot be compact for all $a\in\A$, unless $\A\subset \K(\HH)$. So $(\A,\HH,D_\Lambda )$ in general is not a spectral triple. Nevertheless $[D_\Lambda , a]$ is bounded for any $a\in\A$ and equation \eqref{eq:spectraldistance} still defines an extended metric $d_{\A, D_\Lambda }$ on $\mc{S}(\A)$.

One can also consider the action of $P_\Lambda$ on the algebra. Assuming that $\A$ is unital, let $\pi_\Lambda :\B(\HH)\to\B(\HH)$ be the linear map
\begin{equation}\label{eq:piN}
\pi_\Lambda (a):=P_\Lambda a P_\Lambda
\end{equation}
and $\OO_\Lambda$ the image of $\A^{\mathrm{sa}}$: 
\begin{equation}
\OO_\Lambda:=\pi_\Lambda (\A^{\mathrm{sa}})
\end{equation}

\begin{prop}
\label{prop3.1}
$\OO_\Lambda$ is a finite-dimensional order unit space, with order unit $P_\Lambda$. Every state of $\OO_\Lambda$ is normal,
\begin{equation}
\mc{S}(\OO_\Lambda)=\mc{N}(\OO_\Lambda).
\label{eq:19}
\end{equation}
Furthermore, if
 \begin{equation}\label{eq:30}
L_\Lambda(.):=\|[D_\Lambda , .]\|
\end{equation} is a Lipschitz seminorm, then $(\OO_\Lambda,L_\Lambda)$ is a compact quantum metric space.
\end{prop}
\begin{proof}
Call $\HH_\Lambda:=P_\Lambda\HH$ the range of $P_\Lambda$. Then $\OO_\Lambda$ is a real vector subspace of $\B(\HH_\Lambda)^{\mathrm{sa}}$. Since $1\in\A$, $\OO_\Lambda$ contains the identity operator of $\HH_\Lambda$, that is $P_\Lambda=\pi_\Lambda (1)$. Hence $\OO_\Lambda$ is an order unit space. By \cite[Thm.~4.3.13(ii)]{Kadison1983} every state of $\OO_\Lambda$ can be extended to a state of $\B(\HH_\Lambda)$, hence it is normal, being the latter a finite dimensional matrix algebra.

For the second statement, one can repeat verbatim the proof of \cite[Prop.~4.2]{CDMW09}.
\end{proof}

\begin{rem}\label{footunitization}
If $\A$ is not unital, one may consider its minimal unitization $\A^+ = \A\oplus\C$ with $z\in\C$ acting on $\HH$ as a multiple of the identity operator $1$. $(\A^+,\HH,D)$ may not be a spectral triple (if $D$ has not a compact resolvent, the condition $a(D+i)^{-1}$ is not satisfied). Nevertheless the spectral distance $d_{\A^+,D}$ is still well-defined on $\mc{S}(\A^+)$, and coincides with $d_{\A,D}$ on $\sa\subset\mc{S}(\A^+)$ \cite[Lemma V.4]{MT11}. The same is true for $d_{\A^+, D_\Lambda }$ and $d_{\A,D_\Lambda}$.
\end{rem}

The following example shows the importance of working with ordered unit spaces, rather than only with spectral triples.
\begin{ex}
The complexification $\pi_\Lambda (\A)$ of $\OO_\Lambda$ is a vector subspace of $\B(\HH_\Lambda)$ but not necessarily a subalgebra. For instance take $\HH=\C^4$, $\A\simeq M_2(\C)$ the subalgebra of $M_4(\C)$ of block-diagonal matrices with identical blocks:
\begin{equation}
\hspace{2cm}
\begin{bmatrix}
a_{11} & a_{12} & 0 & 0 \\
a_{21} & a_{22} & 0 & 0 \\
0 & 0 & a_{11} & a_{12} \\
0 & 0 & a_{21} & a_{22}
\end{bmatrix}
\;,\qquad
a_{ij}\in\C \;,
\end{equation}
and $P_\Lambda=\mathrm{diag}(1,1,1,0)$. Every element of $\pi_\Lambda (\A)$ is a matrix with the same element in position $(1,1)$ and $(3,3)$. If $a\in\A$ is
the element with $a_{11}=a_{22}=0$ and $a_{12}=a_{21}=1$, clearly
\begin{equation}
\pi_\Lambda (a)\cdot\pi_\Lambda (a)=
\begin{bmatrix}
1 & 0 & 0 & 0 \\
0 & 1 & 0 & 0 \\
0 & 0 & 0 & 0 \\
0 & 0 & 0 & 0
\end{bmatrix}
\end{equation}
is not in $\pi_\Lambda (\A)$, meaning that $\pi_\Lambda (\A)$ is not a subalgebra of $\B(\HH)$. However $\pi_\Lambda (\A)$ is an algebra as soon as $P_\Lambda$ is in $\A$ or in its commutant $\A'$ (that is $[P_\Lambda,a]~=~0$ for all $a~\in~\A$).$\square$
\end{ex}

\subsection{Truncated topologies}\label{sec:3.2}
Given a unital spectral triple $(\A,\HH,D)$ and a finite-rank projection $P_\Lambda$, we thus obtain three distinct extended metric spaces:
\begin{align}\label{eq:metricspaces}
\big(\mc{S}(\A),d_{\A,D}\big) \,, &&
\big(\mc{S}(\A),d_{\A,D_\Lambda }\big) \,, &&
\big(\mc{S}(\OO_\Lambda),d_{\OO_\Lambda,D_\Lambda }\big) 
\end{align}
where $d_{\OO_\Lambda,D_\Lambda }$ denotes the distance defined by \eqref{eq:orderunitd} on $\mc{S}(\OO_\Lambda)$ by the seminorm $L_\Lambda$. In the passage $d_{\A,D}\to d_{\A,D_\Lambda }$, only the metric structure incoded in $D$ changes. In the passage $d_{\A,D_\Lambda }\to d_{\OO_\Lambda,D_\Lambda }$ the state space itself is modified. We aim at answering two questions:
\begin{itemize}

\item\emph{Equivalence:} Under which conditions are these distances equivalent?
Here the relevant notion is \emph{strong equivalence}: two distances $d_1$ and $d_2$ on a set $X$
are strongly equivalent if there exists positive constants $\alpha,\beta$ such that $\forall\;x,y\in X$:
$$
\alpha d_1(x,y)\leq d_2(x,y)\leq \beta d_1(x,y)\;.
$$

\item\emph{Convergence:}
Given $(\A, \HH, D)$ with infinite dimensional $\A$, can the extended metric space $\big(\mc{S}(\A),d_{\A,D}\big)$  be approximated in a suitable sense by a sequence of metric spaces associated to truncated operators $D_\Lambda $?
\end{itemize}
We study below the first point. The second one is discussed in the next section.

\smallskip

To begin with, let us note that the last distance in \eqref{eq:metricspaces} is defined on a different space than the other two. So in order to compare them one needs to clarify the relation between states of $\A$ and of $\OO_\Lambda$. This is where normal states turn out to be important.

\begin{lemma}\label{prop:3.6}
There is an injective (not necessarily surjective) map $\sharp:\mc{S}(\OO_\Lambda)\to\mc{N}(\A)$, $\varphi\mapsto\varphi^\sharp$, given by:
\begin{equation}
\varphi^\sharp:=\varphi\circ\pi_\Lambda  \;.
\end{equation}
\end{lemma}
\begin{proof}
Since $\pi_\Lambda $ in (\ref{eq:piN}) preserves positivity, and $\varphi^\sharp(1)=\varphi(\pi_\Lambda (1))=\varphi(P_\Lambda)=1$, clearly $\varphi^\sharp$ is a state of $\A$ (it is actually a state of $\A^{\mathrm{sa}}$, extended in a unique way to $\A$ by $\C$-linearity).

Let $\varphi,\psi\in\mc{S}(\OO_\Lambda)$ and suppose $\varphi^\sharp(a)=\psi^\sharp(a)$ for all $a\in\A^{\mathrm{sa}}$. Since $\OO_\Lambda=\pi_\Lambda (\A^{\mathrm{sa}})$, for any $b\in\OO_\Lambda$ there exists $a\in\A$ such that $b=\pi_\Lambda (a)$. One has $\varphi(b)=\varphi^\sharp(a)=\psi^\sharp(a)=\psi(b)$. Hence $\varphi=\psi$, and
the map $\varphi\mapsto\varphi^\sharp$ is injective.

Any $\varphi\in\mc{S}(\OO_\Lambda)$ is normal. Let $R$ be a density matrix for $\varphi$.  For any $a\in\A$ one has $\varphi^\sharp(a)=\tr(RP_\Lambda aP_\Lambda) = \tr(\rho a)$ with $\rho:= P_\Lambda RP_\Lambda$, meaning that $\rho$ is a density matrix for $\varphi^\sharp$, hence the latter is normal.
\end{proof}

The map $\sharp$ being injective allows to define two extended metrics on $\mc{S}(\OO_\Lambda)$:
\begin{equation}
d_{\A,D}^{\,\flat}(\varphi,\psi)
:= d_{\A,D}(\varphi^\sharp,\psi^\sharp) \;\quad\text{ and }\quad
d_{\A,D_\Lambda }^{\,\flat}(\varphi,\psi)
:= d_{\A,D_\Lambda }(\varphi^\sharp,\psi^\sharp).
\end{equation}
In the following proposition we discuss conditions for their equivalence.

\begin{prop}\label{propequiv}
If $L_\Lambda$ is a Lipschitz seminorm, then $\,d_{\A,D_\Lambda }^{\,\flat}\!$ and $d_{\OO_\Lambda,D_\Lambda}\!$ are strongly equivalent on $\mc{S}(\OO_\Lambda)$.
If in addition i) $L_D$ is Lipschitz, or ii) $P_\Lambda$ is in the commutant $\A'$ of $\A$, or iii) $[D,P_\Lambda]=0$, then $d_{\A,D}^{\,\flat}$ and $d_{\OO_\Lambda,D_\Lambda }$ are strongly equivalent.
\end{prop}

\begin{proof}
Since $P_\Lambda$ has finite rank, $\OO_\Lambda\subset\mc{L}^1(\HH)$ and we can consider traceless operators
\begin{equation}
V_\Lambda=\big\{b\in\OO_\Lambda:\tr(b)=0 \big\}.
\end{equation}
By \eqref{eq:28} and \eqref{eq:piN} one has $[D_\Lambda ,P_\Lambda]=0$, thus adding a multiple of the order unit $P_\Lambda$ to $b$ does not change $L_\Lambda(b)$ nor $\varphi(b)-\varphi'(b)$. Therefore for any $\varphi,\varphi'\in\mc{S}(\OO_\Lambda)$:
\begin{subequations}
\begin{equation}\label{eq:3.8a}
d_{\OO_\Lambda,D_\Lambda }(\varphi,\psi)=\sup_{b\in V_\Lambda}\big\{\varphi(b)-\psi(b):L_\Lambda(b)\leq 1\big\} \;.
\end{equation}
Likewise, adding a multiple of $1$ to $a\in\A$ doesn't change $L_D(a)$,$L_\Lambda(a)$ nor \mbox{$\varphi^\sharp(a)-\psi^\sharp(a)$, so}
\begin{align}\label{eq:3.8b}
d_{\A,D}^{\,\flat}(\varphi,\psi)
&= \sup_{a\in\pi_\Lambda ^{-1}(V_\Lambda),\; L_D(a)\leq 1} \; \varphi^\sharp(a)-\psi^\sharp(a)
\; =\sup_{b\in V_\Lambda,\; L'(b)\leq 1}\;\varphi(b)-\psi(b),\\[5pt]
d_{\A,D_\Lambda }^{\,\flat}(\varphi,\psi)
&=\sup_{a\in\pi_\Lambda ^{-1}(V_\Lambda),\; L_\Lambda(a) \leq 1}\;
\varphi^\sharp(a)-\psi^\sharp(a)\; 
=\; \sup_{b\in V_\Lambda,\; L''(b)\leq 1} \; \varphi(b)-\psi(b), 
\label{eq:3.8c}
\end{align}
\end{subequations}
where $\pi_\Lambda ^{-1}(V_\Lambda)$ is the preimage of $V_\Lambda$ in $\A$ and we call
\begin{equation}
L'(b):=\sup_{a\in\pi_\Lambda ^{-1}(V_\Lambda),\, \pi_\Lambda (a)=b} L_D(a)
\;,\qquad
L''(b):=\sup_{a\in\pi_\Lambda ^{-1}(V_\Lambda),\,\pi_\Lambda (a)=b} L_\Lambda(a)
\;.
\end{equation}
The proposition is proved as soon as we show that $L_\Lambda, L'$ and $L''$ are 
norm on $V_\Lambda$: all norms on a finite-dimensional vector space are equivalent,
so that the distances \eqref{eq:3.8a}, \eqref{eq:3.8b} and \eqref{eq:3.8c}
dual to $L_\Lambda$, $L'$ and $L''$ are strongly equivalent.
$L_\Lambda$ being a norm follows from the Lipschitz hypothesis: $L_\Lambda(b)=0$ implies $b=\lambda P_\Lambda$, but
since $b$ is traceless, it must be $b=0$. Regarding $L'$ and $L''$, one easily checks
the triangle inequality, so they are seminorms. 
That $L''$ is actually a norm comes from $[D_\Lambda, P_\Lambda]=0$
that, for any $b =\pi_\Lambda(a)$, implies
\begin{equation}\label{eq:16}
L_\Lambda(a) \geq L_\Lambda(b) \;.
\end{equation}
Hence $L''(b) \geq L_\Lambda(b)$,
and $L''$ is 
a norm too on $V_\Lambda$.

For $L'$, in case $L_D$ is
Lipschitz one has 
\begin{equation}
L'(b)=0, \; b\in V_\Lambda  \Longrightarrow  L_D(a) = 0\quad
\forall a\in\pi_\Lambda^{-1}(V_\Lambda),
\label{eq:17}
\end{equation}
hence $a=\lambda 1$, so that $b=\lambda P_\Lambda =0$ because of the traceless condition. Thus $L'$ is a norm.
 Otherwise one notices that 
\begin{equation}
[D_\Lambda ,\pi_\Lambda(a)]=P_\Lambda\left([D,a]+\big[[P_\Lambda,a],[D,P_\Lambda]\big]\right)P_\Lambda,
\label{eq:67}
\end{equation}
so that $[P,D_\Lambda]=0$ or $P_\Lambda\in\A'$ implies 
\begin{equation}
L_D(a)  \geq L_\Lambda(b) \label{eq:13}
\end{equation}
for all $b=\pi_\Lambda (a)$.  Hence $L'(b)\geq
L_\Lambda(b)$ and $L'$ is a norm. 
\end{proof}

With some conditions on the projection  $P_\Lambda$, but no Lipschitz
condition on $L_D$ nor $L_\Lambda$, one gets that truncating the algebra yields
fewer states (the map $\sharp$ is not surjective on $\sa$) but larger distances.

\begin{prop}
\label{propineq}
For all $\varphi,\psi\in\mc{S}(\OO_\Lambda)$,
\begin{equation}\label{eq:twothree}
d_{\A,D_\Lambda }^\flat(\varphi, \psi)\leq d_{\OO_\Lambda,D_\Lambda }(\varphi,\psi),
\end{equation}
with equality if $P_\Lambda\in\A$. If $[D,P_\Lambda]=0$ or $P_\Lambda\in\A'$, then
\begin{equation}\label{eq:onetwo}
d_{\A,D}^\flat(\varphi, \psi)\leq d_{\OO_\Lambda, D_\Lambda }(\varphi,\psi) \;.
\end{equation}
If $[D,P_\Lambda]=0$ and $P_\Lambda\in\A$, then
\begin{equation}\label{eq:onetwothree}
d_{\OO_\Lambda, D_\Lambda }(\varphi,\psi)  \leq d_{\A,D}^\flat(\varphi, \psi) \;.
\end{equation}
Furthermore, if $P_\Lambda\in\A'$, one also has for all
$\varphi,\psi\in\mc{S}(\A)$:
\begin{equation}
\label{eq3:11}
d_{\A,D}(\varphi,\psi)\leq d_{\A,D_\Lambda }(\varphi,\psi).
\end{equation}
\end{prop}
\begin{proof}
Eq.~\eqref{eq:twothree} is the dual of \eqref{eq:16}:
\begin{align}
d^\flat_{\A,D_\Lambda }(\varphi,\psi) &=\sup_{a\in\A^{\mathrm{sa}}}\big\{ \varphi^\sharp(a)-\psi^\sharp(a): L_\Lambda(a)\leq 1\big\} 
=\sup_{a\in\A^{\mathrm{sa}}}\big\{ \varphi(b)-\psi(b): L_\Lambda(a)\leq 1\big\} \notag\\
&\leq\sup_{b\in\OO_\Lambda}\big\{ \varphi(b)-\psi(b): L_\Lambda(b)\leq 1\big\}
=d_{\OO_\Lambda,D_\Lambda }(\varphi,\psi).
\end{align}
If $P_\Lambda\in\A$, then $\OO_\Lambda\subset\A$ and 
$d_{\A,D_\Lambda }(\varphi^\sharp,\psi^\sharp)\geq
d_{\OO_\Lambda,D_\Lambda }(\varphi,\psi)$, showing the previous inequality is an
equality. 
Eq.~\eqref{eq:onetwo} is the dual of
\eqref{eq:13}. Eq.~\eqref{eq3:11} is the dual of 
\begin{equation}
L_D(a)\geq L_\Lambda(a)\label{eq:68}
\end{equation}
which follows from \eqref{eq:13} and the observation that 
$P_\Lambda\in\A'$ implies $L_\Lambda(a) = L_\Lambda(b)$.

Assume $[D, P_\Lambda]= 0$ and $P_\Lambda\in\A$. Then $L_\Lambda(b) = L_D(b)$ for any
$b\in\OO_\Lambda$ so that
\begin{subequations}
\begin{align}\label{eq:10}
  d_{\OO_\Lambda, D_\lambda}(\varphi, \psi) &=
  \sup_{b\in\OO_\Lambda}\left\{ \varphi(b) - \psi(b) : L_D(b) =
    1\right\}\\
&\leq \sup_{a\in\A}\left\{ \varphi^\sharp(a) - \psi^\sharp(a) : L_D(a) =
    1\right\} =d^\flat_{\A, D}(\varphi, \psi),
\end{align}
\end{subequations}
where we identify $\varphi, \psi$ (defined on the subalgebra
$\pi_\Lambda(\A)\subset \A$) to their extension $\varphi^\sharp,
\psi^\sharp$.
\end{proof}

\noindent
Note that unlike proposition \ref{propequiv}, proposition \ref{propineq} does not require $P_\Lambda$ to be finite rank, not even $D_\Lambda$ to be bounded.

\begin{rem}
When $P_\Lambda$ is a central projection commuting with $D$, then
\eqref{eq:onetwo} and \eqref{eq:onetwothree} combine to give
$d^\flat_{\A, D} = d_{\OO_\Lambda, D_\Lambda}$. This is Lemma 1 of
\cite{Martinetti:2002ij}, that allows to compute the distance in
the spectral triple of the Standard Model,
by reducing $\C\oplus \mathbb{H} \oplus M_3(\C)$ to $\R \oplus \R$.
\end{rem}

\section{Convergence of truncations}

In this section we discuss how elements of $\sa$ can be approximated by sequences of
elements in $\mc{S}(\OO_N)$, as well as the convergence of $(\mc{S}(\OO_N),d^\flat_{\A,D})$
to $(\sa,d_{\A,D})$ in the Gromov-Hausdorff sense.
In \S\ref{sec:3.3} we focus on compact quantum metric spaces, where the equivalence
of the weak* and metric topologies is crucial.
In \S\ref{sec:4.2} we no longer assume this equivalence, and
extend some of the previous results. %to this more general framework.
As shown in \cite{CDMW09}, an example where the two topologies are not equivalent is the Moyal plane.
In \S\ref{sectionlattice} we give another example, the lattice $\Z$:
we prove that, while in the weak$^*$ topology every state can be approximated by truncated states,
in the metric topology this holds true only for states with
finite moment of order $1$ (cf.~Prop.~\ref{conlattice}).

In all this section, $\left\{P_N\right\}_{N\in\N}$ is a sequence of
increasing finite-rank projections\begin{equation}
  \label{eq:90}
  P_N \leq P_{N+1}, \quad \forall N\in\N.
\end{equation}
 We make the
extra-assumption that they converge to $1$ in the weak operator topology, that is
\begin{equation}
\lim_{N\to\infty}(v,P_N w)=(v,w) \quad\text{ for all }v,w\in\HH.
\label{eq:14}
\end{equation}

\subsection[Compact quantum metric spaces \& Gromov-Hausdorff convergence]{Compact quantum metric spaces \& Gromov-Hausdorff convergence}\label{sec:3.3}

Let us begin with a technical lemma.
\begin{lemma}\label{wopconv}
For any $a\in{\mathcal B}(\HH)$, the sequence $\{\pi_N(a)\}_{N\in\N}$ weakly
converges to $a$.
\end{lemma}
\begin{proof}
For any $v,w\in\HH$ one has
\begin{equation}
\inner{v,(P_N aP_N-a)w}=
\inner{v,a(P_N-1)w}
+\inner{(P_N-1)v,aP_N w}.
\end{equation}
Since $\|P_N\|\leq 1$,
\begin{equation}
\left|\inner{v,(P_N a P_N-a)w}\right|\leq 2\|a\|\cdot\left|\inner{v,(P_N-1)w}\right| \;.
\end{equation}
From \eqref{eq:14} it follows $\lim_{N\to\infty}\inner{v,(P_N-1)w}=0$, that concludes the proof.
\end{proof}

For any projection $P_N$, we call ``truncated states'' the image in
  ${\cal  S}(\A)$  of the map $\sharp$ defined in Lemma
  \ref{prop:3.6}.
Any normal state of a unital subalgebra $\A\subset\B(\HH)$ 
can be weakly approximated by a sequence of truncated states.

\begin{prop}\label{prop:3.12}
For any $\varphi\in\mc{N}(\A)$ there is a
sequence of states $\{\varphi_N\}_{N\in\N}$ such that:\\
i) $\varphi_N\in\mc{S}(\OO_N)$ for all $N\geq 0$;
ii) $\varphi_N^{\,\sharp}\to\varphi$ in the weak$^*$ topology. 
\end{prop}

\begin{proof}
Let us choose a density matrix $R$ for $\varphi$ and define  
\begin{equation}
 Z_N:=\tr(RP_N)=\tr(\pi_N (R)) \;.
\end{equation}
Any normal states is continuous on the unit ball of $\B(\HH)$ for the weak operator topology
 \cite[Thm.~7.1.12]{Kadison1983}. 
Since $\|\pi_N (R)\|\leq 1$, Lemma \ref{wopconv} yields
\begin{equation}\label{eq:ZN}
\lim_{N\to\infty} Z_N=1 \;.
\end{equation}
Let $N_\varphi$ denote the smallest integer such that $Z_{N_\varphi}\neq 0$.
Since $\pi_{N_\varphi}\circ\pi_N =\pi_{N_\varphi}$ for all $N\geq N_\varphi$,
we have $Z_N\neq 0$ for all $N\geq N_\varphi$. Thus
\begin{equation}
\varphi_N(b):=Z_N^{-1}\,\tr(Rb)
\label{eq:32}
\end{equation}
is a well-defined state of $\OO_N$ for all $N\geq N_\varphi$. 
For $N<N_\varphi$, choosing arbitrary states
$\varphi_N\in\mc{S}(\OO_N)$ does not modify the limit, and
we do not loose generality assuming $N_\varphi=0$.

We now prove the weak limit.
Due to the linearity of states, it is enough to show it for $\|a\|\leq 1$.
Note that in this case $\|\pi_N(a)\|\leq 1$ too.
By Lemma \ref{wopconv}, $\pi_N(a)\to a$ weakly.
Again by \cite[Thm.~7.1.12]{Kadison1983}, $\tr(R\pi_N(a))\to\tr(Ra)=\varphi(a)$.
Hence by \eqref{eq:ZN}:
\begin{equation*}
\varphi_N^{\,\sharp}(a)=Z_N^{-1}\,\tr(R\pi_\Lambda (a)) \to\varphi(a) \,.\vspace{-25pt}
\end{equation*}
\end{proof}

When $(\A^{\mathrm{sa}},L_D)$ is a compact quantum metric
  space, Prop.~\ref{prop:3.12} shows that any normal state is the
  limit of truncated state in the metric topology induced by $d_{\A,
    D}$. In this case one also has convergence of metric spaces in the Gromov-Hausdorff sense.

\begin{prop}\label{prop:4.3}
Let $(\A^{\mathrm{sa}},L_D)$ be a compact quantum metric space and $\,\overline{\mc{N}(\A)}$
the weak closure of $\,\mc{N}(\A)$. Then $(\mc{S}(\OO_N),d^\flat_{\A,D})$
converges to $(\,\overline{\mc{N}(\A)},d_{\A,D})$ for the Gromov-Hausdorff distance.
\end{prop}

\begin{proof}
Since the map $\sharp:\mc{S}(\OO_N)\to\mc{N}(\A)$ in Prop.~\ref{prop:3.6} is an isometric embedding,
it is enough to prove that the subspaces $X_N:=\sharp\big(\mc{S}(\OO_N)\big)$ of $M:=\sa$ converge
to $\overline{\mc{N}(\A)}$ in the Hausdorff sense.

Since $(\A^{\mathrm{sa}},L_D)$ is a compact quantum metric space, the metric topology coincides with
the weak$^*$ topology on $\sa$. Hence $M$ is compact and $X_N$ are compact subspaces $\forall\;N$.

For a sequence of compact subspaces $\{X_N\}$ of a compact metric space $M$, such that
$X_N\subset X_{N+1}$ for all $N$, the Hausdorff limit $X$ is the closure
of the union $\bigcup_NX_N$ \cite[pag.~253]{BBI01}. %Exercise 7.3.5
Since $X_N\subset\mc{N}(\A)$, then $X\subset\overline{\mc{N}(\A)}$.
On the other hand, from Prop.~\ref{prop:3.12} it follows that $\mc{N}(\A)\subset X$. Hence
$\overline{\mc{N}(\A)}=X$.
\end{proof}

\smallskip

One may wonder what the closure of $\mc{N}(\A)$ is.
If the $C^*$-completion $A$ of $\A$ is a von Neumann algebra, then
$\mc{N}(\A)$ is already closed \cite[Lemma 1]{DellAntonio:1967fk}. As
well, if $A=\K$ then every state is normal and $\mc{N}(\A)$ is closed.

Another important class of examples is given by $\A=C^\infty(M)$ and $\HH=L^2(M)$, with $M$ a compact oriented
Riemannian manifold. In this case $\mc{N}(\A)$ is a proper subset of $\sa$, since for example pure states are not normal.
Nevertheless, it is easy to prove that $\overline{\mc{N}(\A)}=\sa$.
Indeed, let $\psi_{\epsilon,x}$ be the (normalized) characteristic function of the ball with radius $\epsilon$ centered at $x\in M$.
Since $f\in A$ is continuous, $\inner{\psi_{\epsilon,x},f\psi_{\epsilon,x}}\to f(x)$ for $\epsilon\to 0$, and the pure state $\delta_x$
is the weak$^*$ limit of normal states. Hence, $\overline{\mc{N}(\A)}$ contains all finite convex combinations of pure states.
By Krein-Milman theorem \cite{AE80}, every compact convex set (in a locally convex space) is the closure of the convex hull of its extreme points.
Thus, $\sa$ is the closure of finite convex combinations of pure states, and this means $\sa\subset\overline{\mc{N}(\A)}$ (the opposite
inclusion is obvious). The same holds if $\HH=L^2(M,E)$ with $E$ a vector bundle, since $E$ is locally trivial and for $\epsilon$ small
enough we can define a family of sections playing the role of $\psi_{\epsilon,x}$.
Hence,

\begin{cor}
For $\A=C^\infty(M)$, $\HH=L^2(M,E)$ as above, and $D$ a Dirac-type operator,
$(\mc{S}(\OO_N),d^\flat_{\A,D})$ converge to $(\sa,d_{\A,D})$ for the Gromov-Hausdorff distance.
\end{cor}

\subsection[Beyond compact quantum metric spaces: states with finite
$1^\text{st}$ moment]{Beyond compact quantum metric spaces: states with finite
moment of order $1$}\label{sec:4.2}

Let us now consider a spectral triple $(\A, \HH, D)$ such that
$(\A^\text{sa}, L_D)$ is not a compact quantum metric space, and $\{P_N\}_{N\in\N}$ an increasing sequence of
finite rank projections convergent to $1$ in the weak operator
topology. We can always assume that $\A$ is
unital (replacing it by its unitization if needed, as explained in remark~\ref{footunitization}).
The order unit spaces $\OO_N$ are  well defined, and it makes
sense to study the convergence
of the sequence $(\mc{S}(\OO_N),d_{\A,D}^\flat)$.

\smallskip

In the commutative case $\A=C^\infty_0(M)$, $M$ a non-compact Riemannian manifold,  an important class of states regarding the
topology induced by the Wasserstein distance are those with finite
moment of order $1$ (on a connected manifold, the distance between any two such states is
finite). Any $\varphi\in{\cal S}(C_0^\infty(M))$ 
is given by a unique probability measure $\mu$ on $\mc{P}(C_0^\infty(M))\simeq M$, 
\begin{equation}
  \label{eq:9}
  \varphi(f) = \int_M f(x)\;  \de\mu_x \quad \forall f\in C_0(M)
\end{equation}
and
its moment of order $1$ with respect to $x'\in M$ is defined as
\begin{equation}\label{eq:20}
  {\mathcal M}_1(\varphi, x') := \int_M d_{\text{geo}}(x, x') \,\de\mu(x) \;.
\end{equation}
For $M$ connected, the finiteness of ${\mathcal M}_1(\varphi, x')$ does not depend on
the choice of $x'$: either it is finite for all $x'$ or infinite for all $x'$.

In the noncommutative case, 
a similar notion can be defined for states
$\varphi$ that are given by probability measures $\mu$ on the pure
state space, that is such that
  \begin{equation}
\varphi(a) =\int_{\pa} \omega(a) \,\de\mu_\omega \quad \forall a\in\A.
\label{eq:3}
\end{equation}
 This is
not always the case, but there is a large classe of noncommutative
algebras $\A$  for which this does happen (e.g.~unital separable $C^*$-algebras). 
We then say that $\mu$ has a \emph{finite moment of order $1$ with respect to the pure state $\omega'$} if the
expectation of the spectral distance from $\omega'$, viewed as a function on $\pa$, namely
\begin{equation}
  \label{eq:4}
  \mc{M}_1(\mu, \omega') := \int_{\pa} d_{\A, D}(\omega, \omega') \, \de\mu_\omega \;,
\end{equation}
is finite. Notice that unlike the commutative case,
for noncommutative $\A$ there may be different measures $\mu$ on
$\pa$ giving the same state $\varphi$:
the quantity $\mc{M}_1(\mu, \omega')$ (in particular its finiteness) may depend to the choice of
$\mu$, as illustrated in example \ref{exemplematrice} below.

For a normal state $\varphi$, we use the following alternative definition.
Any density matrix $R$ for $\varphi$ is a positive compact operator,
hence it is diagonalizable. Let $\mathfrak{B}=\{\psi_n\}_{n\in\N}$ be an orthonormal basis of
$\HH$ made of eigenvectors of $R$, with eigenvalues $p_n\in\R^+$.
Denote $\Psi_n(a):=\inner{\psi_n,a\psi_n}$  the corresponding vector
states in ${\cal S}(\A)$. Then one has
\begin{equation}\label{eq:6}
\varphi(a) = \sum_{n\geq 0} p_n \, \Psi_n(a) \quad \forall a\in\A .
\end{equation}
\begin{df}
\label{defmoment}
  Given an arbitrary spectral triple $(\A, \HH, D)$ and a density matrix $R$,
  we call moment of order $1$ of $R$ with respect to an eigenbasis $\mathfrak{B}$
  and to a state $\Psi_n$ (induced by a vector $\psi_n\in \mathfrak{B}$)
   the moment of the distribution $\{p_n\}_{n\in\N}$, viewed as a
   discrete probability measure on the lattice, the latter being 
   equiped with cost function $d_{\A,D}$. Explicitly, we define
\begin{equation}
    \label{eq:7}
    {\mathcal M}_1(R,\mathfrak{B},\Psi_n):= \sum_{k\geq 0} p_k \, d_{\A, D}(\Psi_k, \Psi_n).
  \end{equation}
\end{df}

\noindent
It is not difficult to prove that, once fixed $R$ and an eigenbasis $\mathfrak{B}$,
the finiteness of ${\mathcal M}_1(R,\mathfrak{B}, \Psi_n)$ does not depend on the
choice of the vector state $\Psi_n$. 

We stress that definition \ref{defmoment} does not necessarily
coincides with \eqref{eq:4}, because the vector states $\Psi_k$ are not
necessarily pure: they are pure if e.g.~$\A=\K(\HH)$, but they are not
for $\varphi$ a normal state of $\A=C^\infty_0(M)$.
\smallskip

\begin{ex}\label{exemplematrice}
Let $\A=M_2(\C)$. Any pure state $\Psi$ is a vector state, that is
$\Psi(a)=\inner{\psi, a\psi}$ for any $a\in M_2(\C)$, where $\psi$ is a unit vector in $\C^2$ and the inner
product is the usual one. Any two vectors $\psi$ equal up to
a phase determine the same pure state, so that ${\mathcal P}(M_2(\C)$
is the projective space $\C P^1$. The latter is in $1$-to-$1$
correspondence with the sphere $\bS^2$: $\psi_+:=\dbinom{1}{0}$ is mapped to the north pole of  $\bS^2$, $\psi_-:=\dbinom{0}{1}$
to the south pole, and the set of vectors
\begin{equation}
\psi_\theta := \frac{1}{\sqrt 2}\binom{ 1 }{ e^{i\theta} } \qquad
\theta\in [0, 2\pi[
\end{equation}
is mapped to the equator. The whole state space (with weak$^*$-topology) is homeomorphic to the unit ball in $\R^3$.
For instance the center of the ball is the state $\varphi(a):=\frac{1}{2}\tr(a)$. 

We consider the spectral
 triple described in \cite{Iochum:2001fv}, such that the distance
 between any two states is finite if and only if they are at the same
 latitude. In particular
 \begin{equation}
   \label{eq:78}
   d_{\A, D}(\Psi_+, \Psi_-) =\infty, \qquad    d_{\A, D}(\Psi_\theta,
   \Psi_{\theta'}) <\infty \quad \forall \theta, \theta' \in[0, 2\pi[. 
 \end{equation}
The state $\varphi$ has density matrix
$R=\frac{1}{2}\mathbb{I}_2$, meaning that any orthonormal basis of $\C^2$
is an eingenbasis of $R$.
In particular the canonical basis $\mathfrak{B}:=\{\psi_+,\psi_-\}$
and the basis
$\mathfrak{B}_\theta:=\{\psi_\theta,\psi_{\pi+\theta}\}$ for any
value of $0\leq \theta\leq \pi$
yields two distinct decompositions of $\varphi$ on pure states:
\begin{equation}
  \label{eq:26}
  \varphi = \tfrac 12 (\Psi_+ + \Psi_-) = \tfrac 12 (\Psi_\theta + \Psi_{\pi+\theta}).
\end{equation}
Explicitly, for any $a=\left\{a_{ij}\right\}\in M_2(\C)$ one has
\begin{gather}
  \Psi_+(a) =\inner{\psi_+, a\,\psi_+} = a_{11}, \qquad \Psi_-(a)
  =\inner{\psi_-, a\,\psi_-}=a_{22},\\[3pt]
   \Psi_\theta(a) =
    \inner{\psi_\theta, a\,\psi_\theta} = \tfrac12 (a_{11}  + a_{12}e^{i\theta}  +
    a_{21}e^{-i\theta}  + a_{22}).
\label{eq:72}
\end{gather}
Notice that $\Psi_+,\Psi_-$ and $\Psi_\theta, \Psi_{\pi+\theta}$ are pure,
so that here \eqref{eq:26} corresponds to both decompositions
\eqref{eq:3} and \eqref{eq:6}: the first term in \eqref{eq:26} may be
viewed as the discrete measure $\mu:=\left\{p_+=\frac 12, p_-=\frac 12\right\}$ with support on the north and south poles, and the
second term as the discrete measure $\mu_\theta := \left\{p_\theta
  =\frac 12, p_{\theta+\pi} = \frac 12\right\}$ with support on the equatorial points with meridian
coordinates $\theta, \theta+\pi$.

From \eqref{eq:78} we see that moment of order $1$ of
$\varphi$ depends on the choice of the eigenbasis of $R$ (or
equivalently on the choice of the measure): 
\begin{equation}
\mathcal{M}_1(\mu,\Psi_+) = \mathcal{M}_1(R,\mathfrak{B},\Psi_+)=
\frac{1}{2}d_{\A, D}(\Psi_+,\Psi_-) = \infty
\label{eq:76}
\end{equation}
and similarly for $\mathcal{M}_1(\mu,\Psi_-)$,  whereas for any value $\theta\in [0, \pi]$
\begin{equation}
\mathcal{M}_1(\mu_\theta, \Psi_{\theta}) =
\mathcal{M}_1(R,\mathfrak{B}_\theta,\Psi_\theta)=\frac{1}{2}
d_{\A, D}(\Psi_\theta,\Psi_{\pi+\theta})<\infty\label{eq:77}
\end{equation}
and similarly for $\mathcal{M}_1(\mu_\theta, \Psi_{\theta+\pi})$.
\hfill\ensuremath{\square}
\end{ex}

Among the normal states of $\A$, we single out the set $\NN_0(\A)$ of those 
for which there exists at least one density matrix $R$ with an eigenbasis $\mathfrak{B}=\{\psi_n\}$
such that \eqref{eq:7} is finite.

\begin{prop}
Let $\varphi\in\NN(\A)$. For any choice of $(R,\mathfrak{B}, \Psi_n)$ one has
\begin{equation}\label{eq:8}
d_{\A, D}(\varphi,\Psi_n) \leq {\cal M}_1(R, \mathfrak{B},\Psi_n) \;.
\end{equation}
In particular, if $\varphi\in\NN_0(\A)$ then $d_{\A, D}(\varphi,\Psi_n)$ is finite.
\end{prop}

\begin{proof}
From \eqref{eq:6} it follows:
\begin{equation}
\varphi(a) - \Psi_n(a) = \sum\nolimits_k p_k \left(\Psi_k (a)-\Psi_n(a)\right) \leq
L_D(a)\sum\nolimits_k p_k\,d_{\A,D}(\Psi_k, \Psi_n)
\end{equation}
for all $a\in\A^{\mathrm{sa}}$. The last sum is the definition of $\mc{M}_1(R, \mathfrak{B},\Psi_n)$.
\end{proof}

\noindent In a similar way, one obtains that $d_{\A, D}(\varphi,\Psi_n)
\leq \mc{M}_1(\mu,\Psi_n)$ for all choices of $\mu$. In the commutative case one has the equality $d_{\A, D}(\varphi, x')={\cal M}_1(\varphi,x')$ \cite[Prop.~2.2]{DM09}.
In the noncommutative case the equality between spectral distance and moments of order
$1$ defined in \eqref{eq:4} and \eqref{eq:7} does not hold in general, as this would imply that these moments do not
depend on how one decomposes $\varphi$, in contradiction with example \ref{exemplematrice}.

\smallskip

We now prove an analogue of Prop.~\ref{prop:3.12} for spectral triples that are not necessarily compact quantum metric spaces.

\begin{prop}\label{propconvmetric}
Let $(\A, \HH, D)$ be an arbitrary spectral triple and $\left\{P_N\right\}_{N\in\N}$ an
increasing sequence of projections in $\B(\HH)$ convergent weakly to $1$. 
For any $\varphi\in\NN_0(\A)$ such that $\mc{M}_1(R,
  \mathfrak{B},\Psi_n)$ is finite for an eigenbasis $\mathfrak{B}$
  in which the $P_N$'s are all diagonal, then
there exists a sequence $\left\{\varphi_N\right\}_{N\in\N}$
such that: i) $\varphi_N\in {\cal S}(\OO_N)$ for all $N\geq 0$; ii)
$\varphi_N^\sharp \to \varphi$ in the metric topology.
\end{prop}
\begin{proof}

Once fixed $(R, \mathfrak{B})$, the state $\varphi_N$ is defined as in (\ref{eq:32}), namely
\begin{equation}
  \label{eq:34}
  \varphi^\sharp_N(a) = Z_N^{-1} \varphi(\pi_N (a)) = Z_N^{-1}\, \tr\left(P_N RP_N a\right).
\end{equation}
By hypothesis $[P_N, R]=0$ for any
$N\in\N$,  hence 
\begin{equation}
\varphi(a) -\varphi_N^\sharp(a) = \tr\left( \left( 1-P_N\right) R \left( a- \varphi^\sharp_N(a)\right)\right).
\label{eq:12}
\end{equation}
Writing $a_{ij}$ the components of $a$ in the basis $\mathfrak{B}$, one has
\begin{equation*}
  \tr \left(\left(1-P_N\right)Ra\right) = \sum_{n>N} p_n a_{nn} ,\quad \tr\left(
   \left (1-P_N\right) R \varphi_N\left(a\right)\right) =  Z_N^{-1}\sum_{n>N}\left(
    p_n\sum_{k<N} p_k a_{kk}\right)
\end{equation*}
so that 
\begin{equation}
  \label{eq:25}
  \varphi(a) -\varphi^\sharp_N(a) = Z_N^{-1}\sum_{n>N} p_n\left( \sum_{k<N}
    p_k (a_{nn}- a_{kk})\right).
\end{equation}
For any $a$ such that $L_D(a)\leq 1$ one has
$|a_{nn} - a_{kk}| = |\Psi_n(a) - \Psi_k(a)|\leq d_{\A,
  D}(\Psi_n, \Psi_k),$
therefore 
\begin{multline*}
 \sup_{a\in\A, L_D(a)\leq 1} |\varphi(a) -\varphi^\sharp_N(a)| \leq
 Z_N^{-1}\sum_{n>N} p_n\left( \sum_{k<N} p_k \, d_{\A,
  D}(\Psi_n, \Psi_k)\right)\\ \nonumber
\leq Z_N^{-1} \sum_{n>N} p_n\left( \sum_{k<N} p_k \, d_{\A,
  D}(\Psi_n,\Psi_0) + p_kd_{\A,
  D}(\Psi_0,  \Psi_k)\right)\\
\leq Z_N^{-1} \left(\sum_{n>N} p_n \, d_{\A,
  D}(\Psi_n,\Psi_0) + p_n \,{\cal M}_1(R,   \mathfrak{B},\Psi_0)\right).
\end{multline*}
Both terms in the parenthesis are remainders of series converging to
$\mc{M}_1(R,   \mathfrak{B},\Psi_0)$, and so vanish as  $N\to \infty$. 
Since $Z_N\to 1$, one gets $\lim_{N\to\infty} d_{\A, D}(\varphi,
\varphi^\sharp_N) = 0$. 
 \end{proof}

As a corollary, one obtains that any state $\varphi$ in $\NN_0(\A)$ can be approximated in the metric
topology by a sequence of states with finite-rank density matrices.
\begin{cor}
\label{cormetrictop}
Let $(\A, \HH, D)$ be an arbitrary spectral triple and $\varphi\in\NN_0(\A)$. There exists
a sequence $\{\varphi_N\}_{N\in\N}$ of normal states with finite-rank density matrix that
is convergent to $\varphi$ in the metric topology,
\begin{equation}
\underset{N\to\infty}{\lim} d_{\A, D}(\varphi, \varphi^\sharp_N) =0.\label{eq:27}
\end{equation}
Furthermore, for any $\varphi,\varphi'\in \NN_0(\A)$,
\begin{equation}
d_{\A, D}(\varphi,\varphi')=\lim_{N\to\infty}d_{\A, D}(\varphi^\sharp_N,\varphi'^\sharp_N) \;.
\label{eq:29}
\end{equation}
\end{cor}
\begin{proof}
Take $(R, \mathfrak{B})$ such that $\mc{M}_1(R,
  \mathfrak{B},\Psi_0)$ is finite, and $P_N$ the projection on
  the first $N$ vectors of $\mathfrak{B}$. Then (\ref{eq:27}) follows from
 proposition~\ref{propconvmetric}. Eq.~(\ref{eq:29})  comes from the
 $N\to\infty$ limit of the two following equations (obtained by the triangle inequality)
\begin{align*}
  &d_{\A,D}(\varphi,\varphi')\leq d_{\A,D}(\varphi,\varphi_N^\sharp)
  +d_{\A,D}(\varphi',{\varphi'}^\sharp_N) +d_{\A,D}(\varphi^\sharp_N,{\varphi'}^\sharp_N), \\
&d_{\A,D}(\varphi^\sharp_N,{\varphi'}^\sharp_N)\leq d_{\A,D}(\varphi,\varphi^\sharp_N)
+d_{\A,D}(\varphi',{\varphi'}^\sharp_N) +d_{\A,D}(\varphi,\varphi') \;.\label{eq:66}
\end{align*}

\vspace{-.85truecm}\end{proof}

Corollary \ref{cormetrictop} shows that the states with a finite rank density matrix are dense in
$\NN_0(\A)$. There is an important difference with the
situation in the weak$^*$ topology: once fixed the net of projection $P_N$, any normal state
can be weakly approximated by states in ${\cal S}(\OO_N)$. In fact
from Prop.~\ref{prop:3.12}, which is still valid for non-unital algebras (cf.~Rem.~\ref{footunitization}),
one has
\begin{equation}\label{eq:convNA}
\overline{\underset{\longrightarrow}{\lim}\;
   \mathcal{N}(\OO_N)} = \overline{\NN(\A)} \;.
\end{equation}
On a non-compact quantum metric
   space, any state with finite moment of order $1$ can be
   approximated by truncated states, but the truncations  (i.e.~the
   $P_N$'s) depends on the state.
We investigate below simple, a one dimensional lattice, where the
$P_N$'s are actually the same for all states.

\subsection{Example: the lattice $\Z$}\label{sectionlattice}
We identify $\A= C_0(\Z)$ with the algebra of complex diagonal matrices,
\begin{equation}
a=\mathrm{diag}(\ldots a_{-1}, a_0,a_1,\ldots,a_n,\ldots)\quad \text{ with } \quad \lim_{n\to\pm\infty}a_n=0,
\end{equation}
acting on $\HH = l^2(\Z)\otimes \C^2$ as $a\otimes \I_2$.  The
selfadjoint operator $D$ acts on the orthonormal basis
$\mathfrak{B} = \ket{n}_\pm$ of $\HH$ as
\begin{equation}
D\ket{n}_+ =\ket{n+1}_- 
\quad D\ket{n}_- =\ket{n-1}_+ 
\label{eq:15}
\end{equation}
Any state $\varphi$ on $\A$ is a discrete probability distribution
$p = \left\{p_n\in\R^+_0\right\}_{n\in\Z}$,
that is
\begin{equation}
\varphi(a)=\sum\nolimits_{n\in\Z}a_np_n \;.
\label{eq:23}
\end{equation}
It is normal, with density matrix $R=\text{diag}(...,p_{-1},p_0,p_1,...)$.
The state $\Psi_n(a) = a_n$ defined by $\ket n$
is pure, meaning that Def.~\ref{defmoment} and Eq.~\eqref{eq:4} coincide:
${\cal  M}_1(p, \Psi_n)={\cal M}_1(R, \mathfrak{B}, \Psi_n)$. We denote by
${\cal S}_0(\A)$ the set of states with finite moment of order $1$.

\smallskip

The spectral distance $d_{\A, D}$ turns out
to be the Wasserstein distance $W_D$ on $\sa$, introduced in
\cite{Martinetti:2012fk} by
taking as cost function $d_{\A, D}$ on pure
states. Namely
\begin{equation}
 W_D(\varphi,\varphi'):=\sup_{a\in\text{Lip}_D(\A)} 
 \varphi(a)-\varphi'(a),
\label{eq:36}
 \end{equation}
with $\text{Lip}_D(\A)$ the set of elements which are
Lipschitz with respect to the spectral distance:
\begin{equation}
\text{Lip}_D(\A) := \left\{ a\in \A, \,|\Psi(a)-\Psi'(a)|\leq
    d_{\A,D}(\Psi, \Psi') \;
\text{ for all }\;\Psi,\Psi'\in \pa\right\}.
\label{eq:37}
\end{equation}
 \begin{lemma}\label{lemma:Da} For any $a\in\A$, one has
   $\|[D,a]\|=\sup_{n}|a_n-a_{n+1}|$, hence
  \begin{align}
\label{lemmadir2}
    \|[D,a]\|&\leq 1 \quad\text{if{}f }\quad|a_n-a_k|\leq |n-k| \quad \forall
    n,k.
    \end{align}
\end{lemma}
\begin{proof}
Both sides of the first equation are invariant if we add a
constant to $a$, thus we assume that $a_0=0$.
Noticing that $1- D^2$ is the projection operator on $\ket{0}_-$, one has
$\|D\|^2=\|D^2\|=1$. Since $a_0=0$, we have also $[D,a]=D^2[D,a]$. Thus
\begin{equation}
\|[D,a]\|\leq\|D\|\cdot\|D[D,a]\|=\|D[D,a]\| \;.
\end{equation}
On the other hand
\begin{equation}
\|D[D,a]\|\leq\|D\|\cdot\|[D,a]\|=\|[D,a]\|
\end{equation}
which proves $\|[D,a]\|=\|D[D,a]\|$. This norm is easy to
compute since:
\begin{align}
D[D,a]\ket{n}_+ &=(a_n-a_{n+1})\ket{n}_+ \;, &
D[D,a]\ket{n}_- &=(a_n-a_{n-1})\ket{n}_- \;.
\end{align}
This concludes the proof of the equation in the statement of the lemma.

Due to the triangle inequality, $\|[D,a]\|\leq 1$ implies
\begin{equation}
|a_n-a_k|\leq\sum_{j=\min(n,k)}^{\max(n,k)-1}|a_j-a_{j+1}| \leq \sum_{j=\min(n,k)}^{\max(n,k)-1}1=|n-k| \;.
\end{equation}
On the other hand, if $|a_n-a_k|\leq|n-k|$ then $|a_n-a_{n+1}|\leq 1$ and $\|[D,a]\|\leq 1$.
\end{proof}

\medskip
The condition $|a_n-a_k|\leq |n-k|$ is the discrete analogue of the
$1$-Lipschitz condition. This is what makes $d_D$ equal to $W_D$, as
in the continuous case $\A=C_0(M)$ (see e.g.~\cite{DM09}).

\begin{prop}
\label{propMK}  
For any states $\varphi, \varphi'\in S(C_0(\Z))$ one has
\begin{equation}
  \label{eq:54}
  d_{\A, D}(\varphi, \varphi') = W_D(\varphi, \varphi').
\end{equation}
In particular, the spectral distance between any state and the pure
state $\delta_n$, $n\in\Z$ is
 \begin{equation}\label{eq:dooNb}
d_{\A,D}(\varphi,\delta_n)=\sum_{k\in\Z} |k-n|\hspace{1pt}p_k,
\end{equation} meaning that the spectral distance between pure states of the
lattice is 
  \begin{equation}
    \label{eq:53}
    d_{\A, D}(\delta_m, \delta_n) = |m-n|.
  \end{equation}
\end{prop}
\begin{proof}
  By (\ref{lemmadir2}) one has $d_D(\delta_m, \delta_n)\leq
  |m-n|$. The upper bound is attained by the element $a$ with
  components
  $a_i = i$ for $i\leq \sup (m,n)$, zero otherwise.
Hence (\ref{eq:53}). Eq.~ (\ref{eq:54}) follows noticing that
$L_D(a)\leq 1$ is equivalent to $a\in \text{Lip}_D(\A)$.

To prove (\ref{eq:dooNb}), we use again Lemma \ref{lemma:Da} which yields
\begin{equation}
|\varphi(a)-\delta_n(a)|=|\sum\nolimits_{k\in\Z}p_k(a_k-a_n)|\leq\sum\nolimits_{k\in\Z}|k-n|\hspace{1pt}p_k.
\label{eq:62}
\end{equation}
This upper bound is attained by the sequence of elements
\begin{equation*}\label{eq:element}
  a^{(m)}_k=
  \begin{cases}
   k &\text{if}\;k\leq m \;, \\[-1pt]
    2m-k &\text{if}\;m<k\leq 2m \;, \\[-1pt]
    0 &\text{if}\;k>2m \;.
  \end{cases}
\end{equation*}

\vspace{-.85truecm}\end{proof}

\begin{rem}The result for the distance between pure states for the
  finite case ($\A=\C^N$) had been obtained 
\cite{DimakisMuellerHoissen}. Note that the  ``spinorial'' character of the Hilbert space $l^2(\Z)\otimes \C^2$ plays a crucial role. Consider instead $\HH'=l^2(\Z)$, with orthonormal basis $\ket{n}$ and the Dirac operator acting as $D'\ket{n}=\ket{n+1}-\ket{n-1}$. This is a finite approximation of the derivative on $\R$ and has been considered in \cite{BimonteLizziSparano,Atzmon}. In this case the distance between pure states $\delta_n$ and $\delta_m$ is
\begin{subequations}
\begin{align}
d_{\A, D'}(\delta_m, \delta_n) &=  |m-n|+1  &&\text{if}~m-n~\text{is odd}, \\ 
d_{\A, D'}(\delta_m, \delta_n) &= \sqrt{(|m-n|)(|m-n|+1)}  \hspace{-5mm} &&\text{if}~m-n~\text{is even}.
\end{align}
\end{subequations}
\end{rem}

On the lattice, the approximation of a state by its
truncations is always possible in the weak$^*$ topology, but only for states with finite
moment of order $1$ in the metric topology.
\begin{prop}
\label{conlattice}
 In the metric topology induced by $d_{\A,D}$ one has
\begin{equation}
\overline{\underset{\longrightarrow}{\lim}\; \mathcal{S}(\OO_N) }= {\cal S}_0(\A).\label{eq:33}
\end{equation}
In the weak$^*$ topology one has
\begin{equation}
\overline{\underset{\longrightarrow}{\lim}\; \mathcal{S}(\OO_N) }= {\cal S}(\A).\label{eq:33bis}
\end{equation}
\end{prop}
\begin{proof} 
$ {\cal S}_0(\A)\subset \overline{\underset{\longrightarrow}{\lim}\;
  \mathcal{S}(\OO_N) }$ follows from corollary \ref{cormetrictop},
noticing that on the
lattice there is only one eigenbasis $\mathfrak{B}$, hence only one
possible choice of the $P_N$'s. Eq.~(\ref{eq:33}) comes from the
observation that ${\cal S}_0(\A)$ can be equivalently characterized
as the connected component $
\text{Con}(\delta_n)\doteq \left\{\varphi\in \sa, d_{\A,
      D}(\varphi, \delta_n) <\infty\right\}
$ of any pure states $\delta_n$.  As such, it is closed
(and open
as well) for the metric topology \cite[Def.~2.1]{DM09}.

Eq.~\eqref{eq:33bis} follows from \eqref{eq:convNA}, remembering that
$\NN(\A) = {\cal S}(\A)$ and that, for any $C^*$-algebra, ${\cal S}(\A)$
is closed in the weak$^*$ topology.
\end{proof}

The weak$^*$ topology is induced by the distance \cite[Prop.~2.6.15]{BR96}:
\begin{equation}
  \label{eq:31}
  d (R, R') := \|R-R'\|_{\tr}.
\end{equation}
The difference between the weak$^*$ and the metric 
topologies can be seen computing the diameters of the space of
states for the corresponding distances.
\begin{prop}
${\cal S}(\A)$ has infinite diameter for the spectral distance,
diameter $2$ for the metric $d$  inducing the weak$^*$ topology.
\end{prop}
\begin{proof}
For all $\varphi,\varphi'\in {\cal S}(\OO_N)$ and $a$ with $L_D(a)
\leq 1$ we have
\begin{equation*}
\varphi(a) - \varphi'(a) = \sum\nolimits_{n,k=0}^N(a_n-a_k)p_np'_k
\leq
\sum\nolimits_{n,k=0}^N|n-k|p_np'_k  \leq
N\sum\nolimits_{n,k=0}^Np_np'_k=N \;,
\end{equation*}
so $d_{\A,D}(\varphi,\varphi')\leq N$. This upper bound is reached by $d_{\A,D}(\Psi_0,\Psi_N)=N$.
Hence ${\cal S}(\OO_N)$ has diameter $N$ for the spectral
  distance, and from (\ref{eq:33}) ${\cal S}(\A)$ has
  infinite diameter.

For all $\varphi,\varphi'\in {\cal S}(\A)$, one has
\begin{equation}
d(\varphi, \varphi') = \sum_n|p_n-p_n'|\leq \sum_n(p_n+p_n')=2.\label{eq:39}
\end{equation}
The upper bound is reach by $\varphi=\Psi_m$, $\varphi= \Psi_n$
with $n\neq m$.
\end{proof}

\section{Pure states and approximation of points}
\label{sec5}
Having studied in the preceding sections the general topological and metric properties of the various
  truncated distances defined in \eqref{eq:metricspaces}, we now come
  back to the initial motivation of this work, that is understanding what happens to the
  short distance behaviour of a classical  (i.e.~commutative) space
  once a momentum cut-off has been implemented, through the
  substitution of $D$ with $D_\Lambda$. 

Specifically, for  $\A=C_0^\infty(M)$ (as
usual $M$ is an orientable, without boundary, Riemannian manifold), we  study 
how the cut-off in the spectrum of $D$ changes
the topology of the pure state space, i.e.~the points of $M$. We first
consider bounded regularization in \S\ref{subsec5.1}, that is $D_\Lambda$ is a bounded operator with norm $\Lambda>0$. We prove that
the distance $d_{\A, D_\Lambda}$ between two distinct pure states cannot be smaller than
$\Lambda^{-1}$, meaning that the
pure state space $\mc{P}(\A)$ with the metric topology induced by
$d_{\A,D_\Lambda}$ is \emph{not} homeomorphic to $M$ (recall that
$\pa\simeq M$ in the weak$^*$ topology). We investigate the case of
finite rank operator $D_\Lambda$ in \S~\ref{sectionfiniterank}, and prove that any two distinct
pure states are at infinite $d_{\A, D_\Lambda}$ distance.

It is then clear that in a spectral geometry with a cut-off,
points must be replaced by states that are not pure.
In \S\ref{sec:5.3} we investigate the regularization $D_\Lambda=P_\Lambda D$ of the Dirac
operator by its spectral projection $P_\Lambda$ on the interval
$[-\Lambda,\Lambda]$, $\Lambda\in \R^+$.  We
individuate a class of states that are at finite distance, namely the orbits under the geodesic flow of
$D$ of any vector states in the range of
$P_\Lambda$. We stress that this
result is valid for any spectral triple, not
necessarily commutative. Applied to the real line,
it allows to work out states that approximate points in the
 weak$^*$ topology, and whose distance $d_{\OO_\Lambda, D_\Lambda}$ is
 exactly the Euclidean one.
Applications to the circle are the object of \S\ref{sec:6.2.2}.

\medskip

To remain as general as possible, we make $\A=C^\infty_0(M)$
act by pointwise multiplication on the Hilbert space $\HH:=L^2(M,E)$
of square integrable sections of an arbitrary smooth vector bundle
$E\to M$ (not necessarily the spinor or the cotangent bundle), so that
\begin{equation}\label{eq:40}
\|f\|= \sup_{x\in M} |f(x)|.
\end{equation}

\subsection{Bounded regularization}\label{subsec5.1}

We consider regularization by a bounded operator $D_\Lambda$
on $\HH$ (not necessarily with finite rank). $[D_\Lambda,f]$ is clearly bounded and
the spectral distance $d_{\A,D_\Lambda}$ is well-defined.
Borrowing the terminology of \cite{Con94}, the line element ``$ds=D_\Lambda^{-1}$'' is no longer an infinitesimal
(because ~$f(D_\Lambda+i)^{-1}$ is no longer compact for any $f$), so it is reasonable to
expect that points can no longer be taken as close as we want. A
minimum length should appear. From a physical point of view, this
means one cannot
probe the space with a resolution better than $\Lambda^{-1}$ \cite{DFR}.

\begin{prop}\label{prop:1}
Let $D_\Lambda$ be a bounded operator with norm $\Lambda>0$. Then for any $x\neq y$,
\begin{equation}\label{eq:dxy}
d_{\A,D_\Lambda}(\delta_x,\delta_y)\geq\Lambda^{-1} \;,
\end{equation}
i.e.~the distance between two points cannot be smaller than the cut-off.
\end{prop}

\begin{proof}
From (\ref{eq:40}) one has $\|[D,f]\|\leq \|Df\| + \|fD\| \leq 2\|f\|\Lambda.$
Any $f\in\A$ with maximum $f(x)=1/2\Lambda$ and minimum
$f(y)=-1/2\Lambda$ satisfies $\|[D_\Lambda,f]\|\leq 1$ and yields \eqref{eq:dxy}.
\end{proof}

Although the inequality \eqref{eq:dxy} could be trivial (the distance could be simply infinite for all $x\neq y$),
it is still a remarkable result, for it shows that
the extended metric $d_{\A,D_\Lambda}$ and $d_{\A,D}$ are never
strongly equivalent as soon as $D_\Lambda$
is bounded.

\begin{ex}
  A finite distance can be obtained in case $E=M\times\C^2$ (that is
  $\HH=L^2(M)\oplus L^2(M)$) by taking $D_\Lambda:= \Lambda F$ proportional to the flip
  operator 
\begin{equation}
F(\psi_\uparrow\oplus\psi_\downarrow)=\psi_\downarrow\oplus\psi_\uparrow
\qquad\forall\;\psi_\uparrow,\psi_\downarrow\in L^2(M),
\end{equation}
and making $\A$ acts as
$\pi(f)=f\oplus f(x_0)$
where $x_0\in M$ is a fixed base-point ($f$ acts by pointwise
multiplication on the first factor, and through the irreducible
representation $f\mapsto f(x_0)$ on the second). 
Then $d_{\A,D_\Lambda}$ is the discrete metric
\begin{equation}
 d_{\A,D_\Lambda}(\delta_x,\delta_y)= 2\Lambda^{-1} \quad
 \forall x\neq y.
\label{eq:59}
 \end{equation}
Indeed, for any $f\in\A$ one has $ [F,\pi(f)]=((f(x_0)-f)
  \oplus (f- f(x_0)))F.  $ Since $F$ is unitary, one gets
  $ \|[F,\pi(f)]\|=\|f-f(x_0)\|_\infty
  \,
$
(even though the
  representation is sligthly more involved that the pointwise one, \eqref{eq:40} remains valid). Hence
\begin{equation}
|f(x)-f(y)|\leq |f(x)-f(x_0)|+|f(x_0)-f(y)| \leq 2
\end{equation}
for all $f$ such that $\|[F,\pi(f)]\|\leq 1$. This upper bound is attained by any $f$ with
maximum $f(x)=1$, minimum
$f(y)=-1$ and $f(x_0)=0$.
\hfill\ensuremath{\square}
\end{ex}

\subsection{Finite rank regularization}
\label{sectionfiniterank}
We now consider a finite rank operator $D_\Lambda$.
In case $M$ is compact, it can be obtained as
the truncation
$D_\Lambda=P_\Lambda D P_\Lambda$ of the Dirac
operator of $M$ by the action of one of its spectral
projection $P_\Lambda$ ($D$ has compact resolvent, thus
$P_\Lambda$'s have finite rank).
The following results however are valid for arbitrary $M$ and
arbitrary finite-rank operator $D_\Lambda$.

\begin{lemma}\label{lemma:infinity}
Let $P_0$ be a rank $1$ projection and $\psi_0$ a unit vector in the range of $P_0$.
For any $f=f^*\in\A$ one has
\begin{equation}\label{eq:variance}
\|[P_0,f]\|^2=\inner{f\psi_0,f\psi_0}-\left|\inner{\psi_0,f\psi_0}\right|^2 .
\end{equation}
\end{lemma}

\begin{proof}
For any $f\in\A$, call
$f_0:=f-\inner{\psi_0,f\psi_0}$. 
If $f_0\psi_0\neq 0$, we consider the unit vector 
$
\psi_1:=\|f_0\psi_0\|_2^{-1}\,f_0\psi_0.
$
One easily checks that $\inner{\psi_0,\psi_1}=0$, so that $\{\psi_0,\psi_1\}$ is an orthonormal basis
of a $2$-dimensional vector subspace $V$ of $\HH$.
For any $\eta\in\HH$ and $f=f^*$,
\begin{align}
[P_0,f]\eta &=[P_0,f_0]\eta=
\psi_0\inner{\psi_0,f_0\eta}-f_0\psi_0\inner{\psi_0,\eta} 
\label{eqvar2}
=\|f_0\psi_0\|_2(\psi_0\inner{\psi_1,\eta}-\psi_1\inner{\psi_0,\eta}).
\end{align}
Hence
$
[P_0,f]^2=-\|f_0\psi_0\|_2^2 \,\id_V
$ and 
$\|[P_0,f]\|=\|f_0\psi_0\|_2$.
If $f_0\psi_0=0$, one has $[P_0,f]=0$  from \eqref{eqvar2},  and
the previous equation is trivially true.
The lemma follows by linearity of the inner product.
\end{proof}

$\|[P_0,f]\|^2$ is the variance of the random variable $f$ with respect to the
probability measure with density $|\psi_0|^2$ or, in physicists
language,  the uncertainty
$\Delta f$ of the observable $f$ relative to the vector state $\psi_0$.
When working with the spectral distance, it is not uncommon that the
corresponding seminorm is some kind of standard deviation, as recently
stressed by Rieffel \cite{Rie12}.

\begin{prop}
\label{propfiniterank}
Let $D_\Lambda$ be any selfadjoint finite-rank operator on $\HH$. For any $x\neq y$,
\begin{equation}\label{eq:dxyFR}
d_{\A,D_\Lambda}(\delta_x,\delta_y)=\infty \;.
\end{equation}
\end{prop}
\begin{proof}
Using the spectral decomposition of $D_\Lambda$, i.e.~$D_\Lambda=\sum_{n=1}^{r}\!\lambda_nP_n$  where \mbox{$r:=\mathrm{rk}(D_\Lambda)$}
and $P_n:=\psi_n(\psi_n,.)$ are the rank $1$ eigenprojections of $D$,
one obtains from Lemma \ref{lemma:infinity}
\begin{equation}
\|[D_\Lambda,f]\|\leq\sum_{n=1}^{r}|\lambda_n|\|[P_n,f]\| \leq
\sum_{n=1}^{r} |\lambda_n| \|f\psi_n\|.
\label{eq:22}
\end{equation}
We can always find an open neighbourhood $U$ of $x$ with $y\notin U$,
and a real smooth function $f$ with support in $U$ such that $\|f\psi_n\|$ is
as small as we want for any $n\in[1, r]$ and $f(x)$ is arbitrarily
large (take $f$ with a sufficiently narrow peak around $x$). Hence the result.
\end{proof}

\subsection[Regularization by spectral projection \& geodesic flow]{Regularization by spectral projection and geodesic flow}\label{sec:5.3}

The results of the precedent sections indicate that in order to
  have a reasonable topological space associated with the distance
  $d_{\A, D_\Lambda}$, points must be replaced by states that are not pure. This is
 particularly true for finite rank regularizations, as shown in
 proposition \ref{propfiniterank}.

In this section we consider the regularization $D_\Lambda = P_\Lambda
D P_\Lambda$ of the Dirac
operator by its spectral projections $P_\Lambda$ \cite{Reed1980}, and
work out some non-pure states that i) remain at finite distance , ii)
weakly approximate points in the commutative case.
\medskip

Let $(\A, \HH, D)$ be an arbitrary spectral triple. Given $\psi_0\in\HH$ we write $\psi_t :=U_t\psi$
and 
\begin{equation}
\Psi_t(a)=\inner{\psi_t,a\psi_t} \quad\forall\;a\in\A
\end{equation}
the orbits of the vector $\psi_0$ and of the state $\Psi_0$ under the geodesic
flow of $D$ \cite{Reed1980}:
\begin{equation}
U_t~=~e^{itD}, \quad t\in\R.
\label{eq:42}
\end{equation}

\begin{lemma}
Let $\psi_0\in P_\Lambda\HH$. For all $a\in\B(\HH)$, $\Psi_t(a)$ is a differentiable function of $t$ and
\begin{equation}\label{eq:gf2}
\frac{\de}{\de t}\Psi_t(a)=-i\Psi_t([D_\Lambda,a]) \;.
\end{equation}
\end{lemma}
\begin{proof}
 $P_\Lambda$ is the identity operator
on $P_\Lambda\HH$ so that 
$\psi_t=e^{itD}\psi_0 = e^{itD_\Lambda}\psi_0$ and
\begin{align}
\frac{\de}{\de t}\Psi_t(a) &=
\lim_{\tau\to 0}\inner{\psi_t,\frac{(e^{-i\tau D_\Lambda}ae^{i\tau D_\Lambda}-1)\psi_t}{\tau}} \notag\\
&=\lim_{\tau\to 0}\inner{\psi_t,\frac{\left(e^{-i\tau D_\Lambda}-1\right)ae^{i\tau D_\Lambda}\psi_t}{\tau}+
\frac{a\left(e^{i\tau D_\Lambda}-1\right)\psi_t}{\tau}} \notag\\
&=\inner{\psi_t,-iD_\Lambda a\psi_t}+
\inner{\psi_t,iaD_\Lambda\psi_t}=-i\Psi_t([D_\Lambda,a])
\;,
\end{align}
where we use  $\frac{e^{i\tau D_\Lambda} \psi_t-\psi_t}{\tau}\to iD_\Lambda\psi_t$
and $\underset{\tau\to 0}{\lim}\, \psi_{t+\tau}=\underset{\tau\to
  0}{\lim}  e^{i\tau
  D_\Lambda}\psi_t = \psi_t$ \cite[Theo.~VIII.7]{Reed1980}.
\end{proof}

\begin{prop}\label{prop:6.10}
For any  $\psi$ in the range of $P_\Lambda$, the various spectral
distances introduced so far are all  finite on
  any orbit $\Psi_t$ of the geodesic flow of $D$:
\begin{gather}\label{specprojec1}
d_{\A,D_\Lambda}(\Psi^\sharp_{t_1},\Psi^\sharp_{t_2})\leq |t_1 - t_2|, \\[3pt]
d_{\A,D}^\flat (\Psi_{t_1},\Psi_{t_2}) \leq d_{\OO_\Lambda,D_\Lambda}(\Psi_{t_1},\Psi_{t_2}) \leq |t_1 - t_2|. \label{eq:88}
\end{gather}
\end{prop}
\begin{proof}
From \eqref{eq:gf2} one has
$
\Psi_{t_1}(a)-\Psi_{t_2}(a)=i\int_{t_1}^{t_2}\Psi_t([D_\Lambda,a])\de t \;.
$
Since $|\Psi_t(A)|\leq\|A\|$ for any bounded operator $A$,
from Jensen's inequality we get
\begin{equation}\label{eq:6.9}
|\Psi_{t_1}(a)-\Psi_{t_2}(a)|\leq\|[D_\Lambda,a]\|\,\bigg|\int_{t_1}^{t_2}\de t\,\bigg|=
\|[D_\Lambda,a]\|\,|t_1- t_2| \;.
\end{equation}
It is valid for all $a\in\B(\HH)$,  proving both $d_{\A,D_\Lambda}(\Psi_{t_1},\Psi_{t_2})\leq |t_1 - t_2|$ 
and $d_{\OO_\Lambda,D_\Lambda}(\Psi_{t_1},\Psi_{t_2}) \leq |t_1 - t_2|$. Since $[P_\Lambda,D]=0$, %commutes with $D$,
eq.~(\ref{eq:onetwo}) yields $d_{\A,D}^\flat(\Psi_{t_1},\Psi_{t_2}) \leq d_{\OO_\Lambda,D_\Lambda}(\Psi_{t_1},\Psi_{t_2})$.
\end{proof}

We stress that Prop.~\ref{prop:6.10} is true
for arbitrary spectral triples, not necessarily commutative. 
However it is particularly relevant in the commutative case, because $\psi_0$ can be chosen
in such a way that $\Psi_t$ approximates the pure state $\delta_t$. We begin
with the real line and investigate the case of the circle in the next
section.  To
make clear that $t$ is no longer an abstract parameter but a point of space, we switch notation $t\rightarrow x$.

Take
\begin{equation}
\A=C_0^\infty(\R),\quad \HH=L^2(\R),\quad  D=-i\de/\de
x.\label{eq:50}
\end{equation}
Since $\|D_\Lambda\|=\Lambda$, from Prop.~\ref{prop:1} there is a minimum length $\Lambda^{-1}$.
Since $P_\Lambda$ is not of finite rank, Prop.~\ref{propfiniterank} does not apply to this particular example.
Whether $d_{\A,D_\Lambda}(\delta_x,\delta_y)$ is finite or not is still an open problem.

To obtain approximation of points that are at finite distance, we thus
consider the orbit $\left\{\Psi_x =\inner{\psi_x, \cdot
  \psi_x}, x\in\R\right\}$ under the geodesic flow of
$D$ of the state
$\Psi_0=\inner{\psi_0, \cdot\psi_0}$, where $\psi_0$ is a suitably
chosen vector in $P_\Lambda\HH$ as explained in
remark~\ref{rem:approximation} below.
\begin{prop}
\label{prop:realline}
For any $\Lambda$ and $x,y \in \R$ one has
\begin{equation}
  \label{eq:43}
  d_{\A, D_\Lambda}(\Psi_x^\sharp, \Psi_y^\sharp) \leq
  d_{\OO_\Lambda,D_\Lambda}(\Psi_x,\Psi_y)= d_{\A, D}^\flat(\Psi_x, \Psi_y)
  = |x-y| \;.
\end{equation}
\end{prop}
\begin{proof}
The last equality follows noticing that $\Psi_x^\sharp$ is the non-pure state of $C_0^\infty(\R)$ given by
the probability density $|\psi_x|^2$, and $\Psi_y$ is its pull back
under the translation $t\to t+x-y$, namely
\begin{equation}
  \label{eq:61}
  \Psi_x^\sharp(f) = \int_\R f(t) |\psi_x(t)|^2 \de t, \quad \Psi_y^\sharp(f) = \int_\R f(t+x-y) |\psi_x(t)|^2 \de t.
\end{equation}
It is well known that the Wasserstein distance
between translated states on the real line is the amplitude of
translation, that is
$d_{\A,D}^\flat(\Psi_x,\Psi_y)=|x-y|$ (see e.g.~\cite{DM09}). The
thesis then follows from Prop.~\ref{prop:6.10}.
\end{proof}
\smallskip

Viewing the orbit $\R_{\Lambda,\psi_0}:=\{\Psi_x\}_{x\in\R}$
as a ``replica'' of the real line inside
the state space of $C_0(\R)$,  one has that
$(\R_{\Lambda,\psi_0}, d_{\OO_\Lambda,D_\Lambda})$ and
  $(\R^\sharp_{\Lambda,\psi_0}d^\flat_{\A, D})$ (with obvious
  notations) are isometric to
$(\R, |.|)$
for any $\Lambda$ and $\psi_0$.
\begin{rem}
\label{rem:approximation}
  In order that $\Psi_x\to\delta_x$ in the weak$^*$ topology as
  $\Lambda\to\infty$, $\psi_0$ can be taken as the
  Fourier transform of the (normalized) characteristic function of the
  interval $[-\Lambda,\Lambda]$:
\begin{equation}
\psi_0(t)=\frac{1}{2\sqrt{\Lambda\pi}}\int_{-\Lambda}^\Lambda
e^{ipt}\de p=\frac{1}{\sqrt{\Lambda\pi}}\frac{\sin \Lambda t}{t} \;.
\end{equation}
Indeed one then has
\begin{equation}
\Psi_x(f)=\frac{1}{\pi}\int_{-\infty}^\infty \frac{\sin^2
  t}{t^2}\,f(x+\tfrac{t}{\Lambda})\de t \;,
\end{equation}
and clearly $\lim_{\Lambda\to\infty}\Psi_x(f)=f(x)$ for all $f\in\A$
and $x\in\R$.
\end{rem}
 
\subsection{Gromov-Hausdorff convergence on the circle}\label{sec:6.2.2}
We apply some of the previous results to the circle. The spectral triple is
\begin{equation}
\A=C^\infty(\bS^1),\; \HH=L^2(\bS^1,\tfrac{\de x}{2\pi}),\;  D=-i\de/\de x ,
\label{eq:51}
\end{equation}
and we {identify functions on $\bS^1$ with $2\pi$-periodic functions on $\R$.
We use as orthonormal basis of $\HH$ the Fourier modes
\begin{equation}
  \label{eq:52}
  \left\{e_n: x\to e^{in x}, n\in\Z\right\}
\end{equation}
in which $D$ acts as an infinite diagonal matrix, and $f$ as an
infinite matrix with constant diagonals, that is
\begin{equation}\label{eq:98}
\inner{e_n, f e_m}_{L^2}=f_{n-m}
\end{equation}
where $f_n:=\frac{1}{2\pi}\int_{-\pi}^{\pi} f(x)
e^{-inx}dx$ are the Fourier coefficients of $f$ (by
standard Fourier analysis they of rapid decay for $f\in C^\infty(\bS^1)$).

We consider  the regularization by (spectral) projection on the first $N$ positive and negative
Fourier modes. Namely for any $N\in\N$ we write $D_N = P_N D P_N$ where $P_N$ denotes the
projection on
\begin{equation}
\HH_N := \mathrm{Span}\{e_n:|n|\leq N\}.
\end{equation}
The geodesic flow $U_t = e^{t\frac{d}{dx}}$ acts as 
$
U_te_n=e^{int} e_n.
$
Its adjoint action  $\alpha_t(f):~=~U_tfU_t^*$ implements the translation: 
$\alpha_t(f)(x)=f(x+t)$.  
As any autormorphism that preserves the Lipschitz seminorm, $\alpha_t$ is an
  isometry  of the space~of~states.
\begin{lemma}\label{lemma:6.8} Let   $(\A, \HH, D)$  be an arbitrary spectral triple
(not necessarily (\ref{eq:51})). Any automorphism of $\A$ such that $L_{D}(\alpha(a)) = L_D(a) \;\forall a\in \A$ is an
  isometry of the extended metric space $(\mathcal{S}(\A), d_{\A,
    D})$: writing $\alpha^*\varphi := \varphi\circ \alpha$ the pull back of
$\alpha$ on states, one has
  \begin{equation}\label{eq:trasl}
d_{\A, D} (\varphi,\varphi')=d_{\A, D}
(\alpha^*\varphi,\alpha^*\varphi')
\quad \forall \varphi, \varphi'\in{\cal S}(\A).
\end{equation}
In particular, for the spectral triple
(\ref{eq:51}), any translation $\alpha_t,
t\in\R$, and the reflexion 
  $\beta(f)(x):=f(-x)$ are isometries of both $(\sa, d_{\A, D})$
and $(\sa, d_{\A, D_N})$.
\end{lemma}
\begin{proof} Eq.~\eqref{eq:trasl} follows from
  \begin{equation}
    \label{eq:11}
 \sup_{a\in\text{Lip}_D(\A)} \alpha^*\varphi(a) - \alpha^*\varphi'(a) =
   \sup_{\alpha(b) \in\text{Lip}_D(\A)} \varphi(b)- \varphi'(b) = \sup_{b \in\text{Lip}_D(\A)} \varphi(b)- \varphi'(b).
  \end{equation}
For any $t$ the unitary operator $U_t$ commutes with $D$ and
$D_N$, hence $L_D(\alpha_t(a))
= L_D(a)$ for any $a$ and similarly for $L_N$.
The reflexion $\beta$ is implemented by the adjoint action of the
  unitary operator $Ce_n:=e_{-n}$. Since $C$ anticommutes with $D$ and
  $D_N$, one has  $[D,\beta(a)]=
-C[D,a]C$ so that  $L_D(\beta(a)) = L_D(a)$, and similarly for $L_N$. \end{proof}

As approximation of the point $x\in[-\pi, \pi[$, we consider the vector state
$\Psi_{x,N}\in\mathcal{S}(\OO_N)$ defined by the vector in $\HH_N$:
\begin{equation}
\psi_{x,N}:=\frac{1}{\sqrt{N+1}}\sum_{n=0}^Ne^{-inx} e_n.
\label{eq:84}
\end{equation}
\begin{lemma}
\label{lemma5.9} For any $f$ in $C^\infty(\bS^1)$, $\Psi^\sharp_{x,N}(f)$ is the Fej\'er transform of $f$:
  \begin{equation}
\Psi_{x,N}^\sharp (f) 
 = \sum_{n=-N}^N \left(1 - \frac{|n|}{N+1}\right) f_ne^{in x}
=\frac 1{2\pi}\int_{-\pi}^\pi f(t) F_{N+1}(x-t) dt
\label{eq:60}
  \end{equation}
where $F_N(t):=\frac 1N\left( \frac{\sin(Nt/2)}{\sin(t/2)}\right)\!\rule{0pt}{10pt}^2$ is the Fej\'er kernel. 
For any $x\in\bS^1$, the sequence of non-pure states $\left\{\Psi^\sharp_{x,N}\right\}_{N\in\N}$ converges to the pure state $\delta_x$: 
\begin{equation}
\lim_{N\to\infty}\Psi^\sharp_{x,N}(f)=f(x) \qquad \forall f\in C^\infty(\bS^1).
\label{eq:81}
\end{equation}
\end{lemma}
\begin{proof}
By \eqref{eq:98} and \eqref{eq:84} one has
\begin{align*}
\Psi_{x,N}^\sharp(f) &=\inner{\psi_{x,N}, f \psi_{x,N}} =
\frac{1}{N+1}\sum_{n,m=0}^Ne^{i(m-n)x}\inner{e_m,f e_n} 
= \frac{1}{N+1} \sum_{m,n=0}^N e^{i(m-n)x} f_{m-n} \,.
\end{align*}
With some combinatorics, one obtains the r.h.s.~of \eqref{eq:60}:
\begin{multline*}
(N+1)\Psi_{x,N}^\sharp(f)  =\sum_{k=-N}^N f_ke^{ik x }\sum_{\substack{n=0,\ldots,N\\ 0\leq k+n\leq N}} 1
 =\sum_{k=-N}^N f_ke^{ik x }\sum_{n=\max(0,-k)}^{\min(N,N-k)}1
 \\
=\sum_{k=-N}^N  f_ke^{ik x }(\min(N,N-k)-\max(0,-k)+1) =\sum_{u=-N}^N (N+1-|k|)f_ke^{ik x }.
\end{multline*}
Recall that, by induction, the Ces\`aro
sum of a sequence $\left\{a_k\right\}_{k\in\Z}$ is
\begin{equation}
  \label{eq:112}
  \frac 1{N+1} \sum_{n=0}^{N} S_n(a_k) = S_N\left( 1- \frac{|k|}{N+1}a_k\right)
\end{equation}
where $S_n(a_k):= \sum_{k=-n}^n a_k$. Therefore
\begin{equation}
\Psi_{x,N}^\sharp(f) = \frac 1{N+1} \sum_{n=0}^{N} S_n(f_k e^{ikx}).
\label{eq:113}
\end{equation}
This is precisely the Fej\'er transform of $f$, whose integral formula is
given by the second term in \eqref{eq:60} (see e.g.~\cite{various:2008ve}).
 By  Fej\'er theorem, the Fej\'er transforms uniformly converges to $f$ as $N\to\infty$. Thus in particular
$\Psi_{x,N}$ converges to $\delta_x$.
\end{proof}

Since $\psi_{x,N}=U_{-x}\psi_{0,N}$, one has that
\begin{equation}
\bS^1_{N}=\big\{\Psi_{x, N}^\sharp: x\in \bS^1\big\}
\label{eq:75}
\end{equation}
is the orbit of $\Psi_{0,N}$ under the geodesic flow. On the real line, both $d_{\OO_\Lambda, D_\Lambda}$ and
$d^\flat_{\A, D}$  on similar orbits coincides with the geodesic distance (Prop.~\ref{prop:realline}). For the
circle the same is true only in the $N\to\infty$ limit. 

\begin{prop}
\label{lemma:dgeo}
i) For all $x,y\in \bS^1$, one has
\begin{equation}\label{eq:ineqcircle}
d_{\A,D}^\flat(\Psi_{x,N},\Psi_{y,N}) \leq d_{\OO_N,D_N}(\Psi_{x,N},\Psi_{y,N}) \leq d_{\mathrm{geo}}(x,y),
\end{equation}
and
\begin{equation}\label{eq:limitcircle}
\lim_{N\to\infty}d_{\A,D}^\flat(\Psi_{x,N},\Psi_{y,N})
=\lim_{N\to\infty}d_{\OO_N,D_N}(\Psi_{x,N}, \Psi_{y,N})=d_{\mathrm{geo}}(x,y).
\end{equation}
ii) A lower bound for $d^\flat_{\A, D}$ and an alternative upper bound are provided by
\begin{equation}
\rho_N( x - y ) \leq d^\flat_{\A,D}(\Psi_{ x ,N},\Psi_{ y ,N})\leq \rho'_N( x - y )
\label{eq:64}
\end{equation}
where
\begin{align}
\label{defrho}
\rho_N(x) &:=
\frac{8}{\pi}\sum_{\substack{1\leq n\leq N\\[1pt] n\;\mathrm{odd}}}\frac{(-1)^{(n-1)/2}}{n^2}\left(1-\frac{n}{N+1}\right)\sin\frac{nx}{2} \;,\\
\rho'_N(x) &:=
2\sqrt 2 \sqrt{\sum_{1\leq n\leq N}   \frac 1{n^2} \left(1-\frac{n}{N+1}\right)^2 \left(\sin \frac{nx}{2}\right)^2} \;.
\label{eq:65}
\end{align}
\end{prop}
\begin{proof}
Since $\alpha^*_t\Psi_{x,N}=\Psi_{x+t,N}$ and $\beta^*\Psi_{x,N}=\Psi_{-x,N}$, by Lemma \ref{lemma:6.8} it is enough to
prove the proposition for $y=0$ and $0<x\leq\pi$.

i) Eq.~\eqref{eq:ineqcircle} follows from Prop.~\ref{prop:6.10} noticing that for $x\in [0,\pi]$
the geodesic distance is $d_{\mathrm{geo}}(x,y)=|x-y|$.
Since $(\A^{\mathrm{sa}},L_D)$ is a compact quantum metric space (see e.g.~the introduction of \cite{Rieffel:1999ec}), the weak$^*$ topology and metric topology of $d_{\A,D}$ coincide. By Lemma \ref{lemma5.9} one has $\Psi_{x,N}\to\delta_x$ and so $\lim_{N\to\infty}d_{\A,D}(\Psi_{x,N},\Psi_{y,N}) = d_{\mathrm{geo}}(x,y)$,
proving \eqref{eq:limitcircle}.

ii) Note that $d^\flat_{\A,D}(\Psi_{ x ,N},\Psi_{ 0 ,N})=d^\flat_{\A,D}(\Psi_{x/2,N},\Psi_{-x/2,N})$,
that by \eqref{eq:60} is the sup of
\begin{equation}\label{eqrhoprime}
\Psi_{x/2,N}^\sharp(f)-\Psi_{-x/2,N}^\sharp(f)=
-4\sum_{1\leq n\leq N} \left(1 - \frac{n}{N+1}\right)\, \mathcal{I}\mathrm{m}(f_n)\, \sin(\tfrac{n}{2}x) \;,
\end{equation}
over real $1$-Lipschitz functions $f$ (we used $f_{-n} = f^*_n$).

The periodic $1$-Lipschitz function defined by $f(t)=|t|$ for $t\in [-\pi,\pi]$ (so, the geodesic distance of the circle)
has Fourier coefficients $f_0 = \frac{\pi}{2}$, $f_n = f_{-n} = -\frac{2}{n^2\pi}$ for $n$
odd, and $f_n=0$ for even $n\neq 0$. Translating this function by $\pi/2$ amounts to rescaling $f_n$ by $e^{in\pi/2}$,
which is equal to $i(-1)^{(n-1)/2}$ for $n$ odd and gives the lower bound \eqref{defrho}.

On the other hand, from \eqref{eqrhoprime}, for any $1$-Lipschitz $f$:
\begin{equation}
\Psi_{x/2,N}^\sharp(f)-\Psi_{-x/2,N}^\sharp(f)
\leq 4\sum_{1\leq n\leq N} \left(1 -\frac{n}{N+1}\right) |f_n\sin(\tfrac{n}{2}x)|\leq 2\sqrt 2 \sqrt{\sum\nolimits_{n=1}^N C_n^2}
\end{equation}
where
\begin{equation}
C_n: = \frac 1n \left(1 -\frac{n}{N+1}\right)\sin\frac{nx}{2}
\end{equation}
and we used Cauchy-Schwarz inequality together with Parseval identity:
\begin{equation}
2\sum_{n=1}^N|nf_n|^2=\sum_{n=-N}^N|nf_n|^2\leq\sum_{n=-\infty}^\infty |nf_n|^2=\frac 1{2\pi}\int_{-\pi}^{\pi}
  |f'(t)|^2 dt \leq 1.
\end{equation}
This proves \eqref{eq:65}.
\end{proof}

\medskip

\begin{figure}[t]
\begin{center}

\includegraphics[width=7cm]{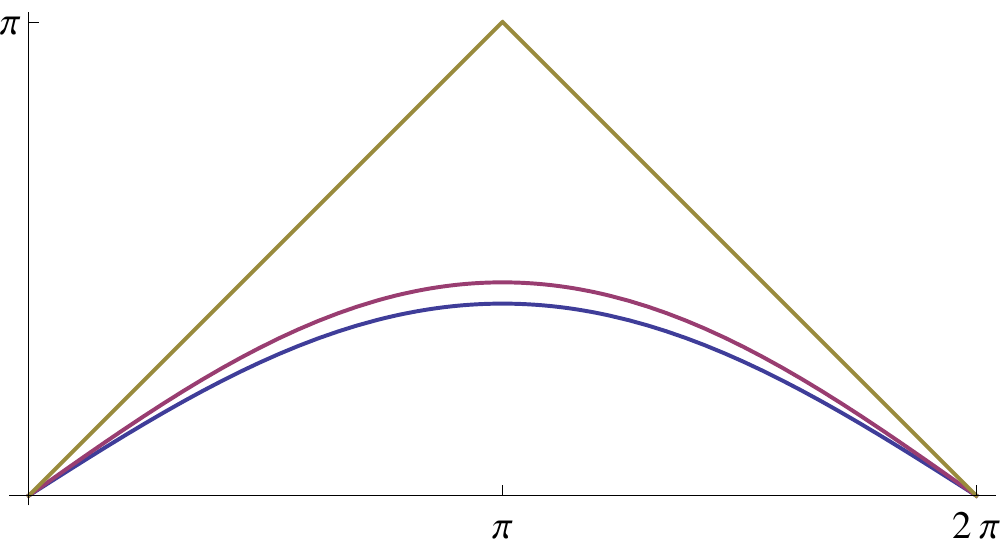}\hspace{10pt}
\includegraphics[width=7cm]{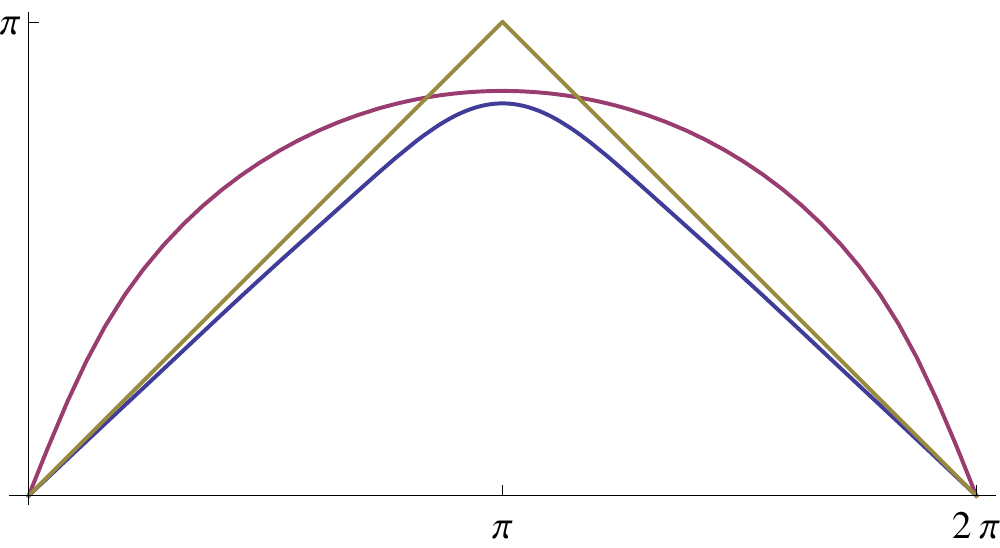}

\parbox{12cm}{\caption{\sl The lower and upper bounds $\rho_N(x),
\rho'_N(x)$. For $N=1$ (on the left), $\rho_N'(x)$ is a better
bound for $d_{\A, D}(\Psi_{x,N}, \Psi_{0,N})$ than $|x|$ (straight line) for
any $x$. For $N\geq 1$ (e.g.~$N=5$ on the right side), $\rho_N'(x)$ is a better
bound than $|x|$ only for $x$ near $\pi$.}\label{figurerho}}

\vspace{-5mm}

\end{center}
\end{figure}

Eq.~\eqref{eq:64} shows that at fixed $N$, none of truncated distances
actually equals the geodesic distance (see Figure~\ref{figurerho}). Unlike the real line, 
the orbit $\bS^1_N$ is not a replica of the
circle with its geodesic distance. However the sequence $\left\{\bS^1_N\right\}_{N\in\N}$ converges to it in the
Gromov-Hausdorff distance.

\begin{prop} 
\label{prop:circle}
$(\bS^1_N,d_{\A,D})$ converges to $(\bS^1, d_{\text{geo}})$ for the Gromov-Hausdorff distance.
\end{prop}

\begin{proof}
From \cite[pag.~253]{BBI01}, %Exercise 7.3.4
the Gromov-Hausdorff limit of $\bS^1_N$ is the set $X$ of
limits of all convergent sequences $\{\varphi_n\in
\bS^1_n\}_{n\in\N}$.
But $(C^\infty(\bS^1),
L_D)$ is a compact quantum metric space \cite{rieffel2003} and by Lemma \ref{lemma5.9} $\Psi_{x,N}\to\delta_x$ in the weak$^*$ topology, thus $\Psi_{x,N}\to\delta_x$
also in the metric topology of $d_{\A, D}$. Hence
$\bS^1\subset X$.

Consider a convergent sequence $\{\varphi_n\in \bS^1_n
\}_{n\in\N}$. By definition, for any $\varphi_n\in\bS^1_n$ there exists $x_n\in\bS^1$ such that $\varphi_n=\Psi_{x_n,n}^\sharp$. 
Let $N'>N$ and consider the Lipschitz periodic function defined by $f(t):=|t|$ for $t\in[-\pi,\pi]$.
By easy calculation one obtains
\begin{align*}
\Psi_{0,N'}^\sharp (f)-\Psi_{x,N'}^\sharp (f)
 &= \frac{8}{\pi}\sum_{\substack{1\leq n\leq N'\\[1pt] n\;\mathrm{odd}}} \frac{1}{n^2}\left(1-\frac{n}{N'+1}\right)\left(\sin\frac{nx}{2}\right)^2
 \geq\frac{4}{\pi}\left(\sin\frac{x}{2}\right)^2=\rho_1(x) \;,\\
\Psi_{0,N}^\sharp(f)-\Psi_{0,N'}^\sharp(f) &\geq\frac{4}{\pi}\sum_{\substack{N<n\leq N'\\[1pt] n\;\mathrm{odd}}}\frac{1}{n^2}\left(1-\frac{n}{N+1}\right)\geq 0 \;,
\end{align*}
with $\rho_1(x)$ as in \eqref{defrho} (the first inequality follows
from the observation that we have a sum of positive terms -
  hence superior to the term $n=1$ - and $\frac{N'}{N'+1}\geq\frac{1}{2}$ if $N'\geq 1$; in the second we used $-\frac{1}{N'+1}\geq-\frac{1}{N+1}$). Hence
\begin{equation}
d^\flat_{\A,D}(\Psi_{0,N},\Psi_{x,N'}) \geq 
\big\{\Psi_{0,N'}^\sharp(f)-\Psi_{x,N'}^\sharp(f)\big\}+\big\{\Psi_{0,N}^\sharp(f)-\Psi_{0,N'}^\sharp(f)\big\}
\geq \rho_1(x) \;,
\end{equation}
and by translation invariance:
\begin{equation}
\label{eq:rhoN1}
\rho_1( x - y )\leq d^\flat_{\A,D}(\Psi_{ x ,N},\Psi_{ y ,N'}) \;.
\end{equation}
Now, since $\mathcal{S}(C^\infty(\bS^1))$ is complete,
$\{\varphi_n\}$ is a Cauchy sequence for the spectral distance. Hence
(\ref{eq:rhoN1}) shows that for all $\epsilon>0$ there exists $N\geq 1$ such that
\begin{equation}
\rho_1(x_m -x_n) \leq \;d_{\A,D}(\varphi_m,\varphi_n)<\epsilon\quad \forall\;m>n\geq N.
\end{equation}
This means
$
|x_m-x_n| <2\arcsin\sqrt{\tfrac{\pi}{4}\epsilon},
$
proving that $\{x_n\}$ is a Cauchy sequence, thus convergent.
Let $x:=\lim_{n\to\infty}x_n$. From the triangle inequality and Prop.~\ref{lemma:dgeo} we get
\begin{equation}
d_{\A,D}(\varphi_n,\delta_x)\leq 
d_{\A,D}(\Psi_{x_n,n}^\sharp,\Psi_{x,n}^\sharp)+d_{\A,D}(\Psi_{x,n}^\sharp,\delta_x)
\leq
d_{\mathrm{geo}}(x_n, x)+d_{\A,D}(\Psi_{x, n}^\sharp,\delta_x)
 \;.
\end{equation}
Since $\Psi_{x,n}^\sharp\to\delta_x$ in the weak$^*$ topology,
$d_{\A,D}(\Psi_{x,n}^\sharp,\delta_x^\sharp)\to 0$. On the other hand
$x_n\to x$, hence $d_{\mathrm{geo}}(x_n, x)\to 0$ too.
Thus
$
\lim_{n\to\infty}d_{\A,D}(\varphi_n,\delta_x)=0
$
proving that $\varphi_n\to\delta_x$ converges to a pure state, and $X\subset\bS^1$. This concludes the proof.
\end{proof}

%%% ======================================================================

\section{Wasserstein distance and Berezin quantization}\label{sectionBerezin}

We now discuss an application of the truncation procedure to quantum spaces.
Given a spectral triple $(\A,\HH,D)$ and a projection $P\in{\cal B}(\HH)$ such that 
$P\cdot\mathrm{dom}(D)\subset\mathrm{dom}(D)$, we denote by
$\A_P$ the algebra generated by the elements $\pi(a):=PaP$, with $a\in\A$. Let $D_P: =PDP$ be the truncated Dirac operator. Note that unlike section \ref{sec5}, we do not
assume that $D_P$ is bounded. 

Many well known noncommutative spectral triples $(\A_P, P\HH,
D_P)$ are obtained in this way, that is by the action of a
  projection $P$ on a commutative spectral triple $(\A,\HH,D)$: Moyal plane, fuzzy spaces, quantum discs,
and more generally any Berezin-Toeplitz quantization of a classical
space. Specifically we study the quantization of the plane in \S\ref{sec:7.2.1},
and of the sphere in \S\ref{sectionfuzzysphere}. Before that, we give in \S\ref{sec:7.1} an
application to gauge theory.

\smallskip

Having in mind the analogy between the spectral and the
Wasserstein distances, a first general result is that a quantum transport is more expensive than a classical
transport. \begin{lemma}
\label{lemma:qAA} 
There is a map $\sharp:\mc{S}(\A_P )\to\mc{S}(\A)$, $\varphi\mapsto\varphi^\sharp$, given by
\begin{equation}
\varphi^\sharp=\varphi\circ\pi \;.
\label{eq:102}
\end{equation}
If $[P,D]=0$ or $[P,a]=0$ for all $a\in\A$, then for any
$\varphi,\psi\in\mc{S}(\A_P )$ one has
\begin{equation}
d_{\A_P , D_P }(\varphi,\psi) \geq d_{\A,D}(\varphi^\sharp,\psi^\sharp).
\label{eq:69}
\end{equation}
\end{lemma}
\begin{proof}
Since $\pi$ preserves positivity of operators, $\varphi^\sharp$ is positive.
When $\A_P$ is unital, with unit $P$,  then
$\varphi^\sharp(1)= \varphi(P) =1$ and $\varphi^\sharp$ is a state. 
Otherwise, $\varphi$ extends in a unique way to the unitization of
$\A_P$, whose unit is the identity on $P\HH$, that is $P$. Hence this
extension satisfies  $\varphi(P)=1$, and as before $\varphi$ is a
state.
 The proof of (\ref{eq:69}) is the same as in Prop.~\ref{propineq}.
\end{proof}
\begin{rem}
The map
  $\sharp$ is not necessarily injective,  unlike Prop.~\ref{prop:3.6}, because in general $\pi$ is
  not surjective.
\end{rem}

\subsection{Gauged Dirac operators}\label{sec:7.1}

The Dirac operator $D$ of a Riemannian spin manifold
$M$ can be lifted to any vector bundle $E\rightarrow M$ as follows. By Serre-Swan theorem, the
set $\cal E$ of smooth  sections of $E$ vanishing at
infinity is a finitely generated projective $\A$-module, with $\A =
C^\infty_0(M)$: namely ${\cal E} \simeq P\A^n$ for some $n\geq 1$ and some
projection $P\in M_n(\A)$. Let $\HH =L^2(M,\mathcal{S})$ be the space of square integrable
spinors and  $\HH^n=\HH\otimes\C^n$. We denote $D_P =
P(D\otimes \I_N)P$ the lift of
$D$ to $E$. It acts on $\HH^n$ and is well defined, because $P$ being
smooth sends the domain of $D\otimes \I_N$ into itself.  In gauge theories, $E$ is a  $SU(n)$ bundle  describing a fermion
paired with a $\mathfrak{su}(n)$ gauge field \cite{Lan02}, and $D_P$ is
then called ``gauged Dirac operator''. 

\begin{lemma}
\label{lemmagauge}
One has $\A_P\simeq\A$.
\end{lemma}
\begin{proof}
Since $\A$ is commutative and $P\in M_n(\A)$, $\pi(a):=P(a\otimes\I_n)P$ is a representation of $\A$ and $\A_P \simeq\A/\ker\pi$.
Let $k=\tr(P)$ be the matrix trace of $P$ (not the trace on $\HH$). This should be an element of $\A$, but in fact
it is an integer, since it coincides with the rank of the
vector bundle $E$. Since $\tr(\pi(a))=ka$, $\pi(a)=0$ implies $a=0$ and $\ker\pi=\{0\}$.
\end{proof}
One can prove that
$(\A, P\HH^n , D_P )$ is a spectral triple. Indeed, the construction
described here is very common in index theory, because for $M$ an
even-dimensional manifold, the Fredholm
index of $D_P$ gives an integer-valued pairing between $D$
(or more generally, a $K$-homology class for $M$) and the class of $E$
in $K^0(M)$ (see \cite{Con94, Mos97}).

\smallskip

The spectral triple $(\A, \HH, D)$ is metrically equivalent to
$(\A,\HH^n,D\otimes\I_n)$, where the algebra acts diagonally on $\HH^n$.
By Lemma \ref{lemmagauge} we identify $\A$ with $\A_P$
and $\varphi$ with~$\varphi^\sharp$. Lemma \ref{lemma:qAA} then tells us that
``gauging'' a Dirac operator makes distances larger.
\begin{prop}
$d_{\A, D_P }(\varphi,\psi) \geq d_{\A,D}(\varphi,\psi) \qquad\forall\;\varphi,\psi\in\sa$.
\end{prop}

Another ways to lift $D$ to $E$, that works with arbitrary connections on $E$, is by the so-called \emph{fluctuation of the metric} \cite{AC2M2}.
%This works with any connection and let the Hilbert space $\HH$ unchanged.
In this way, one obtains a covariant Dirac operator whose associated spectral distance strongly depends on the holonomy of the
connection, as shown in \cite{Martinetti:2006db, Martinetti:2008hl}.

\subsection{Berezin quantization of the plane}\label{sec:7.2.1}

We first recall how to quantize $\C \simeq\R^{2}$
by projecting
$\HH=L^2(\C,\frac{\de^2z}{\pi})$, with inner~product 
\begin{equation}
\inner{f,g}_{L^2}:=\frac{1}{\pi}\int_{\C}\overline{f(z)}g(z)\de^2z,
\label{eq:85}
\end{equation} 
on a suitable Hilbert subspace. Then we prove a result similar to Lemma~\ref{lemma:qAA}.

\smallskip

Fix a real deformation parameter $\theta>0$ and denote by
$\HH_\theta\subset\HH$ the Hilbert subspace spanned by the set of
orthonormal functions
\begin{equation}\label{eq:BM1}
h_n(z):=
\frac{z^n}{\sqrt{\theta^{n+1}n!}}\,e^{-\frac{|z|^2}{2\theta}} \quad n\in\N.
\end{equation}
Notice that $\HH_\theta$ is isomorphic to the holomorphic Fock space
$L^2_{\mathrm{hol}}(\C,e^{-\frac{|z|^2}{\theta}}\frac{\de^2z}{\pi\theta})$ 
via the module map $f\mapsto \tilde f: =\sqrt{\theta}\,e^{\frac{1}{2\theta}|z|^2}\!f$.
Let $P_\theta$ be the orthogonal projection $\HH\to\HH_\theta$, namely
\begin{equation}
P_\theta:=\sum_{n=0}^\infty  h_n\inner{h_n,\,.\,}_{L^2}.
\label{eq:44}
\end{equation}
Two maps are naturally associated to it \cite{BS06}:
the \emph{Toeplitz quantization} $\pi_\theta$ from bounded continuous functions $f$ to bounded operators on $\HH_\theta$:
\begin{equation}
\pi_\theta(f):=P_\theta fP_\theta \;,
\label{eq:103}
\end{equation}
and the \emph{Berezin symbol} $\sigma_\theta(T)$, defined for $T\in\B(\HH)$ by
\begin{equation}\label{eq:BM0}
\sigma_\theta(T)(z)=\Psi_z(T) \;,
\end{equation}
where $\Psi_z(T)=\inner{\psi_z, T\psi_z}_{L^2}$ is the vector state defined by the unit vector in~$\HH_\theta$:
\begin{equation}
\psi_z := e^{-\frac{|z|^2}{2\theta}}\sum\nolimits_{n=0}^\infty
\frac{{\bar z}^n}{\sqrt{\theta^n n!}}\,h_n \;.
\label{eq:93}
\end{equation}
Their composition $B_\theta:=\sigma_\theta\circ\pi_\theta$ is called \emph{Berezin transform}.
Both maps $\pi_\theta$ and $\sigma_\theta$ (and then $B_\theta$), 
are unital, positive and norm non-increasing, that is: $\|\pi_\theta(f)\|\leq\|f\|_\infty$ and
$\|\sigma_\theta(T)\|_\infty\leq\|T\|$ for all $f,T$ (the latter simply follows from $\Psi_z$ being a state).

Let us now consider the canonical spectral triple of $\C\simeq\R^2$, given by:
\begin{equation}
\A=\mc{S}(\R^2) \;, \qquad \HH\otimes\C^2 \;,\qquad  D=2\mat{ 0 & \!\!\!-\bar\partial \\ \partial & \;0 } \;,
\label{eq:105}
\end{equation}
where $\mc{S}(\R^2)$ is the algebra of Schwartz functions on the
plane. We write $z=x+iy$,
so that the derivatives are $\partial=\frac{1}{2}(\partial_x-i\partial_y)$ and $\bar\partial=\frac{1}{2}(\partial_x+i\partial_y)$.
Let $\OO_\theta$ be the order unit space spanned by $\pi_\theta(f)$,
$f\in\A^{\mathrm{sa}}$, and $\pi_\theta(1)=P_\theta$. The action of $P_\theta$ on
$D$, with a proper normalization factor\footnote{The factor $2$ for $D_\theta$ is required so that  $d_{\A,D}(\delta_z,\delta_{z'})=|z-z'|
=\|(x-x',y-y')\|$ coincides with the geodesic distance on the plane.}, yields the Dirac operator of the irreducible spectral triple of Moyal plane \cite{MT11,DML13}.
\begin{prop}
\label{prop:MoyalBerezin}
The truncated Dirac operator $D_\theta:=2(P_\theta\otimes\I_2) D (P_\theta\otimes\I_2)$ is given by:
\begin{equation}
D_\theta=\frac{2}{\sqrt\theta}\mat{ 0 & \mf{a}^\dag\! \\ \mf{a} & 0\, }
\label{eq:96}
\end{equation}
where $\mf{a}^\dag, \mf{a}$ are the creation, annihilation operators:
$\mf{a}^\dag h_n=\sqrt{n+1}\,h_{n+1}$, $\mf{a}h_{n}=\sqrt{n}\, h_{n-1}$.
\end{prop}
\begin{proof}
One has $\bar\partial
e^{-\frac{|z|^2}{2\theta}}=-\frac{z}{2\theta}e^{-\frac{|z|^2}{2\theta}}$,
so that
$
\bar\partial h_{n}=-\frac{ \mf{a}^\dag }{2\sqrt{\theta}}h_n.
$
Thus $P_\theta\bar\partial P_\theta=-\frac{ \mf{a}^\dag
}{2\sqrt{\theta}}$ and by conjugation
$P_\theta\partial P_\theta=\frac{ \mf{a}}{2\sqrt{\theta}}$. Hence \eqref{eq:96}.
\end{proof}

Although we work with Schwartz functions, the quantization map makes sense for more general (even unbounded)
functions. In particular one has $\pi_\theta(z)=\sqrt{\theta}\mf{a}^\dag$
(coming from $zh_n(z)=\sqrt{\theta(n+1)}h_{n+1}$), and by
conjugation $\pi_\theta(\bar z)=\sqrt{\theta}\mf{a}$. Hence
\begin{equation}
[\pi_\theta(\bar z),\pi_\theta(z)]=\theta \;.
\end{equation}
In other terms, cutting-off the Euclidean plane with $P_\theta$ yields a canonical quantization of the plane.\footnote{As usual, the commutation relation holds on a dense subspace of $\HH_{\theta}$ containing the linear span of the basis elements $h_n$.}
The map $\sharp:\mc{S}(\OO_\theta)\to\mc{S}(\A)$ in \eqref{eq:102} is injective, because $\pi_\theta$ in \eqref{eq:103} is surjective by construction, and maps ``quantum states'' into ``classical states''. Even though $P_\theta$ and $D$ do not commute, we are able to obtain in Prop.~\ref{propBerezin} below a result similar to lemma
\ref{lemma:qAA}, together with an upper bound for $d_{\OO_\theta,D_\theta}$ given by the distance
\begin{equation}
d_{\A,D}^{(\theta)}(\varphi,\psi) :=\sup_{f=f^*\in\A}
\big\{ \varphi(f)-\psi(f)\,:\, \|[D,B_\theta(f)]\|\leq 1 \big\}
\;\quad \forall \varphi,\psi\in\mc{S}(\A).
\label{eq:109}
\end{equation}

We begin with two technical lemmas.

\begin{lemma}
\label{technicallemmamoyal}
For any $f\in\A$ and $T\in\OO_\theta$, one has{\footnote{Inside the commutator with $D$, we identify an
element $f$ of $\A$ with
its representation $f\otimes \I_2$ on $\HH\otimes \C^2$.}}
\begin{equation}\label{eq:Dtpt}
[D_\theta,\pi_\theta(f)\otimes\I_2]=(\pi_\theta\otimes \I_2)([D,f])
\,,\qquad [D,\sigma_\theta(T)]=(\sigma_\theta\otimes \I_2)([D_\theta,T]) \,.
\end{equation}
\end{lemma}
\begin{proof} The rank $1$ projection in direction of $\psi_z$,
\begin{equation}
Q_z:=e^{-\frac{|z|^2}{\theta}}\sum_{m,n\geq 0}
\frac{{\bar z}^mz^n}{\sqrt{\theta^{m+n}m!n!}}\,h_m\!\inner{h_n,\,.\,}_{L^2},
\end{equation}
is the density matrix of the coherent state
$\Psi_z$.
With some computations one verifies that
\begin{equation}\label{lemma:BM0}
\partial Q_z+\frac{1}{\sqrt{\theta}}[\mf{a},Q_z]=0 \;,\qquad
-\bar\partial Q_z+\frac{1}{\sqrt{\theta}}[\mf{a}^\dag,Q_z]=0 \;.
\end{equation}
Using the explicit form of the Toeplitz operator,
\begin{equation}\label{eq:berezinexplicite}
\pi_\theta(f)=\sum_{m,n\geq 0}h_m\inner{h_m,fh_n}_{L^2}\inner{h_n,\,.\,}_{L^2}
\equiv\frac{1}{\pi\theta}\int_{\C}f(z)Q_z\de^2z \;
\end{equation}
one obtains, after integration by part:
\begin{equation}
[\mf{a}, \pi_\theta(f)] = \sqrt \theta \,\pi_\theta(\partial
  f), \qquad [\mf{a}^\dag, \pi_\theta(f)] = -\sqrt \theta
  \,\pi_\theta(\bar\partial f).\label{eq:74}
  \end{equation}
Hence the first equation in \eqref{eq:Dtpt}.

From \eqref{eq:BM0} one has $\sigma_\theta(T)(z) = \text{Tr}(Q_z
T)$. Together with \eqref{lemma:BM0} this yields
\begin{equation}
  \label{eq:79}
  \partial \sigma_\theta(T)_{\lvert_z} = -\frac 1{\sqrt \theta}
  \text{Tr}\left([\mf{a}, Q_z] T\right) = -\frac 1{\sqrt \theta}
  \text{Tr}\left(Q_z[T,\mf{a}]\right) =\frac 1{\sqrt \theta}
  \sigma_\theta\left([\mf{a},T]\right)_{\lvert_z} \,.
\end{equation}
Similarly $\bar\partial \sigma_\theta(T)_{\lvert_z} = -\frac 1{\sqrt \theta}
  \sigma_\theta\left([a^\dagger,T]\right)(z)$. Hence the second equation  in \eqref{eq:Dtpt}.
\end{proof}

\begin{lemma}
$B_\theta$ is a self-adjoint automorphism of the vector space $\mc{S}(\R^2)$.
\end{lemma}
\begin{proof} The set $\mc{S}(\R^2)$ is a pre-Hilbert space with inner
  product $\inner{\,,\,}_{L^2}$.
Introducing the \emph{reproducing kernel},
\begin{equation}\label{eq:106}
K_z(\xi)=\theta^{-1}\left|\inner{\psi_\xi, \psi_z}\right|^2=\theta^{-1}e^{-\frac{1}{\theta}|z-\xi|^2},
\end{equation}
one obtains the integral form of the Berezin transform
\begin{equation}
B_\theta(f)(z)=\frac{1}{\pi}\int_{\C}K_z(\xi)f(\xi)\de^2\xi=\inner{K_z,f}_{L^2} \;.
\label{eq:87}
\end{equation}
Since $e^{-\frac{1}{\theta}|z-\xi|^2}$ is a Schwartz function, and the Schwartz space
is closed under convolution, $B_\theta(f)\in\A$ for all $f\in\A$. With a simple explicit
computation one verifies that $\inner{f,B_\theta(g)}_{L^2}=\inner{B_\theta(f),g}_{L^2}$
for all $f,g\in\A$, that is the Berezin transform is self-adjoint.

In Fourier space, the Berezin transform becomes the pointwise
multiplication of the Fourier transform $\hat f$ of $f$ by
a Gaussian (the Fourier transform of the Gaussian kernel
$K_z$). This is identically zero if and only if  $\hat
f = 0$, i.e.~only iff $f=0$.
This proves injectivity.

Let $V=B_\theta(\A)$ and $V^\perp$ its orthogonal complement in the Hilbert space closure of $\mc{S}(\R^2)$.  For any $f\in
V^\perp$ one has
$0=\inner{f,B_\theta(g)}_{L^2}=\inner{B_\theta(f),g}_{L^2}$ for all
$g$. Choosing
$g=B_\theta(f)$, one proves that $B_\theta(f)$ (hence $f$ by
injectivity) vanishes. Thus
$V^\perp=\{0\}$ and $B_\theta$ is surjective.
\end{proof}

\begin{prop}\label{propBerezin}
For all $\varphi,\psi\in\mc{S}(\OO_\theta)$:
\begin{equation}
d_{\A,D}(\varphi^\sharp,\psi^\sharp)\leq
d_{\OO_\theta,D_\theta}(\varphi,\psi)\leq
d_{\A,D}^{(\theta)}(\varphi^{\sharp},\psi^{\sharp}) \;.
\label{eq:104}
\end{equation}
\end{prop}
\begin{proof}
$\sigma_\theta$ and
$\pi_\theta$ are norm non-increasing. Moreover, being $\pi_\theta$
surjective,  there is always an $f$ such that $T=\pi_\theta(f)$, that is $\sigma_\theta(T) = B_\theta(f)$.  Omitting the identity, \eqref{eq:Dtpt} yields
\begin{equation}
\|[D_\theta,\pi_\theta(f)]\|\leq \|[D,f]\| \;, \qquad
\|[D,B_\theta(f)]\|\leq \|[D_\theta,T]\|
\end{equation}
The opposite inequalities for the dual distances then follow.
\end{proof}

Let us apply these results to the coherent states $\Psi_z$, $z\in\C$, 
that are the states of $\OO_\theta$ defined by the vectors $\psi_z$ in \eqref{eq:93}.
Recall that coherent states are the ``best
approximation'' of points in a quantum context, because they minimize
the uncertainty of $\hat{z}\hat{z}^*$ (the square of the ``distance
operator'' \cite{DFR,MMT11,Bahns}), where $\hat{z}=\pi_\theta(z)$. Another way to see that coherent states are a good approximation
of points is to notice that $\Psi_z^\sharp$ is the (non-pure) state of the Schwartz
  algebra $\A={\cal S(\C)}$ given by the evaluation at $z$ of the
  Berezin transform, 
  \begin{equation}
\Psi^\sharp_z(f)=B_\theta(f)(z).
\label{eq:107}
\end{equation}
As such $\Psi^\sharp_z$ converges to the pure state $\delta_z$ as $\theta\to 0$,
as follows from: 
\begin{lemma}\label{lemma:BM2} For any $f\in\A$, one has
$
\|f-B_\theta(f)\|_\infty\leq\sqrt{\pi\theta}\, L_D(f).
$
\end{lemma}
\begin{proof}
We use \eqref{eq:87}.
For any $z\in\C$, $\pi^{-1}K_z(\xi)$ is a Gaussian probability measure on $\C$.
Similarly to \cite[Theo.~2.3]{Rie04b}, one has ($r:=|\xi|$):
\begin{gather*}
|f(z)-B_\theta(f)(z)| =\frac{1}{\pi}\left|\int_{\C}K_z(\xi)\big(f(z)-f(\xi)\big)\de^2\xi\right|
\leq \frac{1}{\pi}\int_{\C}K_z(\xi)\big|f(z)-f(\xi)\big|\de^2\xi
\\[2pt]
\leq L_D(f) \frac{1}{\pi}\int_{\C}|z-\xi|K_z(\xi)\de^2\xi 
=L_D(f) \int_0^\infty \frac{2}{\theta}e^{-\frac{1}{\theta}r^2}r^2\de r
=\sqrt{\pi\theta} L_D(f) \,.\qedhere
\end{gather*}
\end{proof}

Prop.~\ref{propBerezin} allows to compute the distance between coherent
states,  and retrieve a result proved in \cite{MT11}  from a completely different perspective.
\begin{prop}
\label{cor:coherent}
For any $z, z'\in\C$, 
\begin{equation}\label{eq:83}
d_{\OO_\theta, D_\theta}(\Psi_z, \Psi_{z'}) = |z -z'| \,. 
\end{equation}
\end{prop}
\begin{proof}
Due to \eqref{eq:104}, it is enough to prove that
\begin{equation}
d_{\A,D}(\Psi_z^\sharp, \Psi_{z'}^\sharp)=|z-z'| = d_{\A,D}^{(\theta)}(\Psi_z^\sharp,\Psi_{z'}^\sharp).
\label{eq:55}
\end{equation}
By \eqref{eq:107} and \eqref{eq:87}, $\Psi_z^\sharp$ and $\Psi_{z'}^\sharp$ are the non pure states of $\A$ given
by the Gaussian measures $\pi^{-1}K_z$, $\pi^{-1}K_{z'}$. The first equality in \eqref{eq:55} would
be immediate if $\A$ were the algebra $C^\infty_0(\R^2)$: in this case $d_{\A,D}$ would be the Wasserstein
distance on the Euclidean plane, and it is known that the distance between two Gaussians (with the same variance) is the
Euclidean distance between the peaks (see e.g.~\cite{DM09}). 

Here $\A=\mc{S}(\R^2)$
is smaller than $C^\infty_0(\R^2)$, so in principle $d_{\A,D}(\Psi_z^\sharp, \Psi_{z'}^\sharp)\leq
|z-z'|$, but one easily checks that the supremum is attained on the sequence of
$1$-Lipschitz Schwartz functions (of the variable $\xi$):
\begin{equation}\label{eq:seq}
f_{n,z,z'}(\xi)=\bar\xi\cdot\frac{z-z'}{|z-z'|}\,e^{-\frac{1}{n}\,|\xi|^2} \;.
\end{equation}
On the other hand, since $B_\theta$ is surjective on $\A$,
\begin{align*}
d_{\A,D}^{(\theta)}(\Psi_z^\sharp, \Psi_{z'}^\sharp)
&=\sup_{f=f^*\in\A}\big\{ B_\theta(f)(z)-B_\theta(f)(z') \,:\, \|[D,B_\theta(f)]\|\leq 1 \big\}
\\[3pt]
&=\sup_{g=g^*\in\A}
\big\{ g(z)-g(z') \,:\, \|[D,g]\|\leq 1 \big\} \leq |z-z'| \;,
\end{align*}
where $g=B_\theta(f)$ and, since $\A\subset C^\infty_0(\R^2)$, the
distance above is no greater than the geodesic distance. Using again the
sequence \eqref{eq:seq} one proves that the supremum is attained on Schwartz
functions, and last inequality is in fact an equality.
\end{proof}

\begin{rem}
Let $\inner{A,B}_{\mathrm{HS}}=\theta\,\tr(A^*B)$ be the Hilbert-Schmidt inner product on $\mc{L}^2(\HH_\theta)$.
By an explicit computation one checks that
$\inner{\pi_\theta(f),\pi_\theta(g)}_{\mathrm{HS}}=\inner{f,B_\theta(g)}_{L^2}$ for all $f,g\in\A$.
Since $\|B_\theta(f)\|_\infty\leq\|f\|_\infty$,
\begin{equation}
\|\pi_\theta(f)\|_{\mathrm{HS}}^2\leq \|f\|_{L^1}\|f\|_\infty
\end{equation}
is finite, proving that elements of $\pi_\theta(\A)$ are Hilbert-Schmidt operators.
If we replace $\OO_\theta$ by the algebra $\A_\theta$ generated by $\pi_\theta(\A)$,
we get a spectral triple $(\A_\theta,\HH_\theta\otimes\C^2,D_\theta)$ close to the irreducible
spectral triple of Moyal plane (for the latter, one uses the algebra of rapid decay matrices,
in the basis $h_n$, that is dense in the algebra of Hilbert-Schmidt operators). Note that this spectral triple is metrically equivalent to the one in~\cite{GGISV04}. 
\end{rem}
%%% ======================================================================

\subsection{Quantum discs, fuzzy spaces and other examples}

Berezin quantization  applies as well to the unit disc
$\mathfrak{D}=\{z\in\C:|z|<1\}$. With measure  
\begin{equation}
\de\mu=\frac{1}{(1-|z|^2)^2}\de^2z \;,
\end{equation}
one projects on the subspace
$\F:=L^2_{\mathrm{hol}}(\mathfrak{D},\de\mu)$ of holomorphic functions
(the Bergman space) of the Hilbert space $\HH=L^2(\mathfrak{D},\de\mu)$. This yields the quantum disc of
\cite{KL92}.
Using instead the measure
\begin{equation}
\de\mu_\alpha=\pi^{-1}(\alpha+1)(1-|z|^2)^\alpha\de^2z \quad -1<\alpha<\infty,
\end{equation}
the corresponding spaces $\F_\alpha:=L^2_{\mathrm{hol}}(\mathfrak{D},\de\mu_\alpha)$ are the \emph{weighted Bergman space}
(the Hardy space if $\alpha=0$). The corresponding truncated algebra describes the quantum disc
of \cite{GQBV06}.
With a more complicated measure (not absolutely continuous with respect to the Lebesgue one)
one get the $q$-disc \cite[Eq.~(3)]{Kli04}.
Another example, where the projection operator has finite rank, is
given by the fuzzy disk \cite{LVZ03}.
\smallskip

A similar construction holds for the torus $\T^2$ and the projective space $\CP^n$.
For $E\to\T^2$, resp. $E'\to\CP^n$, holomorphic line bundles with Chern number $N$, the subspaces of
$L^2(\T^2,E)$, resp. $L^2(\CP^n,E)$, of holomorphic sections is
finite-dimensional with dimension $|N|$, resp. $\binom{n+|N|}{n}$.
In both cases the corresponding projection has finite rank so that the
truncation yield a finite-dimensional spectral triple:
for instance for $n=1$, $\CP^1\simeq\bS^2$ and one gets the fuzzy sphere
(see e.g.~\cite{SS11}). 

A more general class of examples is given by the Berezin quantization of a compact K{\"a}hler manifold,
that is always given by finite-dimensional full matrix algebras. The quantization map $\pi_N:\A\to\A_N$ is surjective (so $\pi_N(\A)$ is already an algebra), and provides a strict deformation quantization in the sense of Rieffel \cite{BMS94}.
\medskip

Among all these examples, we study in the following the fuzzy
sphere (and keep the other examples for further works). Clearly the
Dirac operator of $\CP^n$ does not commute with
the projection (for example, $P \partial/\partial\bar z_i P=0$ for $z_i$ a homogeneous coordinate
on $\CP^n$ or the complex coordinate on the covering $\C$ of $\T^2$),
so that Lemma \ref{lemma:qAA} does not apply.
However, it is possible to obtain the fuzzy $\CP^n$ (and more generally fuzzy homogeneous spaces) using
projections that commute with the Dirac operator. For $\CP^n$ one
projects on a finite direct sum of irreducible
representations of $SU(n+1)$, the so called \emph{Weyl-Wigner
formalism} (see \cite{VG89}). At
least for the fuzzy sphere, this gives rise to the same quantized
space as Berezin quantization. 
A third way to obtain fuzzy spaces is via coherent states
quantization, that we investigate in the next section.
Metric properties of the fuzzy sphere are investigated in \cite{DLV12},
where we derive that the distance between coherent state converges to the
geodesic distance in the $N\to\infty$ limit. In this paper we show in Prop.~\ref{ineqfuzzysphere} and  Prop.~\ref{prop:fuzzysphere} the existence of an upper and lower bounds for the truncated states, and that also in this case pure states are at an infinite distance.

\medskip

Before that, let us point out that the regularization of the real line by spectral
projection, investigated in \S\ref{sec:5.3}, is an example of Berezin
quantization. Indeed consider the spectral triple \eqref{eq:50}, with $P_\Lambda$ the spectral projection of $D$ in the interval
$[-\Lambda,\Lambda]$. By  Fourier transform one proves that:
\begin{equation}
P_\Lambda=\int_{-\infty}^\infty\de t\,\frac{\sin \Lambda t}{\pi t}\,U_{t} \;,
\end{equation}
where $U_tf(x)=f(x+t)$.
Let $K_x$ be the following kernel:
\begin{equation}
K_x(t):=\frac{\sin\Lambda(x-t)}{\pi(x-t)} \;.
\end{equation}
For any $x$, $K_x$ is a vector in $P_\Lambda\HH$ (with norm $\sqrt{\Lambda/\pi}$). One has
$
P_\Lambda(f)(x)=\inner{f,K_x}
$
hence for any $f\in\HH_\Lambda$, $\inner{f,K_x}=f(x)$.
In other words, $P_\Lambda\HH$ is a reproducing kernel Hilbert space, and the usual
cut-off procedure on the real line is yet another example of Berezin
quantization.
%%% ======================================================================
\subsection{The fuzzy sphere as a coherent state quantization}
\label{sectionfuzzysphere}

The standard Berezin quantization of the sphere $\bS^2\simeq\CP^1$ consists in
taking a power of the quantum line
bundle (i.e.~the dual of the tautological bundle)
and project on the finite-dimensional space of
holomorphic sections. Here we follow the alternative approach of
coherent state quantization.

The canonical spectral triple for $\bS^2$ is $(C^\infty(\bS^2) , L^2(\bS^2)\otimes\C^2,D)$ where
\begin{equation}
D=\bigg(\!\begin{array}{cc}
\frac{1}{2}+\partial_H & \partial_F \\ \partial_E & \frac{1}{2}-\partial_H
\end{array}\!\bigg) 
\label{eq:110}
\end{equation}
and in spherical coordinates $\phi\in[0,2\pi]$ and $\vartheta\in[0,\pi]$ the derivatives in \eqref{eq:110} are:
\begin{equation}
\partial_H = - i \,\frac{\partial}{\partial\phi} \;,
\qquad
\partial_E = e^{ i \phi} \biggl( \frac{\partial}{\partial\vartheta} +  i  \cot\vartheta \frac{\partial}{\partial\phi} \biggr) \;,
\qquad
\partial_F = - e^{- i \phi} \biggl( \frac{\partial}{\partial\vartheta} -  i  \cot\vartheta \frac{\partial}{\partial\phi} \biggr) \;.
\label{eq:94}
\end{equation}
We write the inner product on $L^2(\bS^2)$ as
\begin{equation}
\inner{f,g}_{L^2}=\int_{\bS^2}\overline{f(x)}g(x)\de\mu_x\label{eq:99}
\end{equation}
with $\de\mu_x$ the $SU(2)$-invariant measure normalized to $1$.
An orthonormal basis is given by Laplace spherical
harmonics $Y_{\ell,m}$.{\footnote{Within our normalization, one has e.g.~$Y_{0,0}(x)=1$
and not $1/\sqrt{4\pi}$ as more commonly used.}}

The notation $\partial_H, \partial_E, \partial_F$ comes from the fact that
these operators are the image of the standard Chevalley generators $H=H^*$, $E$ and $F=E^*$ of
the Lie algebra $\mathfrak{su}(2)$ under the representation $\partial:\mathfrak{su}(2)\to\mathrm{Der}(C^\infty(\bS^2))$
as vector fields on $\bS^2$. Let us recall that, in the Chevalley basis, the defining relations
of $\mathfrak{su}(2)$ are $[E,F]=2H$, $[H,E]=E$, $[H,F]=-F$.
The irreducible representation $\rho_\ell:\mathfrak{su}(2)\to\mathrm{End}(V_\ell)$ with highest $\ell\in\frac{1}{2}\N$
is defined as follows;
the underlying vector space $V_\ell\simeq\C^{2\ell+1}$ has orthonormal basis $\ket{\ell,m}$, with $m=-\ell,\ldots,\ell$, and
\begin{equation}
 \rho_\ell(H)\ket{\ell,m} =m\ket{\ell,m} \;,\qquad
 \rho_\ell(E)\ket{\ell,m} =\sqrt{(\ell-m)(\ell+m+1)}\ket{\ell,m+1} \;,
\end{equation} 
with $\rho_\ell(F)=\rho_\ell(E)^*$. The representation $\partial$ decomposes as direct sum of all
$\rho_\ell$ with integer $\ell$, and the equivalence
\begin{equation}
U\partial_\xi U^*=\oplus_\ell
\rho_\ell(\xi)\;\;\forall\;\xi\in\mf{su}(2)
\label{eq:100}
\end{equation}
is implemented by the unitary map
\begin{equation}
U: L^2(\bS)\to \K:= \bigoplus_{\ell\in\N}V_\ell,\quad 
U(Y_{\ell,m}):=\ket{\ell,m}.
\label{eq:108}
\end{equation}

The irreducible spectral triple{\footnote{As for Moyal plane, we
  use a index notation instead of $\qA$,$\qH$,$\qD$ to stress the
  $\ell$-dependence of the objects.}}
  $(\A_\ell,V_\ell\otimes\C^2,D_\ell)$ of the fuzzy sphere is obtained by the
action of the orthogonal projection $P_\ell: \K\to V_\ell$ on
$(C^\infty(\bS^2),\K, UDU^*)$, which is unitary equivalent to the
canonical spectral triple of the sphere. Namely \cite[eq.~(4.1)]{DLV12} 
\begin{equation}
\A_\ell:=\mathrm{End}(V_\ell)\simeq M_{2\ell+1}(\C), \quad D_\ell:=P_\ell (UDU^*) P_\ell=\bigg(\!\begin{array}{cc}
\frac{1}{2}+\rho_\ell(H) & \rho_\ell(F) \\ \rho_\ell(E) & \frac{1}{2}-\rho_\ell(H)
\end{array}\!\bigg) \;.
\label{eq:111}
\end{equation}

We equip $\A_\ell$ with the Hilbert-Schmidt inner product:
\begin{equation}
\inner{A,B}_{\mathrm{HS}}:=\gamma_\ell^{-1}\tr(A^*B) \;\text{ with }\;  \gamma_\ell:=2\ell+1.
\end{equation}
The \emph{covariant Berezin symbol} $\sigma_\ell:\A_\ell\to\A$ is defined
as
$\sigma_\ell(a)(x):=
\Psi_{x;\ell}(a),
$ where 
$\Psi_{x;\ell}$
is the \emph{Bloch coherent state}. Namely $\Psi_{x;\ell}(a)=\tr(Q_{x;\ell}\, a)$  is the vector state of $\A_\ell$
defined by the rank $1$ projection
\begin{equation}
Q_{x;\ell}=
\sum_{m,n=-\ell}^{\ell}\binom{2\ell}{\ell+m}^{\frac{1}{2}}
\binom{2\ell}{\ell+n}^{\frac{1}{2}}e^{i(n-m)\phi}(\sin\tfrac{\vartheta}{2})^{2\ell+m+n}(\cos\tfrac{\vartheta}{2})^{2\ell-m-n}\ketbra{\ell,m}{\ell,n} \;.
\end{equation}
One easily checks that for any fixed $\ell$, the map
$\bS^2\to\mc{S}(\A_\ell)$, $x\mapsto\Psi_{x;\ell}$,
is injective.
We denote by $\pi_\ell$ the adjoint map:
\begin{equation}
\inner{f,\sigma_\ell(a)}_{L^2}=\inner{\pi_\ell(f),a}_{\mathrm{HS}}
\quad\forall f\in \A,\,a\in\A_\ell.
\label{eq:95}
\end{equation}
\begin{rem}
The operator $\pi_\ell(f)$ is not the Toeplitz operator $\breve{\pi}_\ell(f):=P_\ell UfU^*P_\ell$
given by the action of the projection, as in Moyal
case. It is an easy exercise to check that:
\begin{equation}
\breve{\pi}_\ell(f)=\gamma_\ell\int_{\bS^2}f(x)R_{x;\ell}\de\mu_x \;,
\end{equation}
where $R_{x;\ell}$ is the rank $1$ projection in the direction of the vector
$\sum_{m=-\ell}^\ell \overline{Y_{\ell,m}(x)}\ket{\ell,m}$,~while
\begin{equation}\label{eq:BSpil}
 \pi_\ell(f)=\gamma_\ell\int_{\bS^2}f(x)Q_{x;\ell}\de\mu_x \;.
 \end{equation}
That the quantization maps $\pi_\ell$ and $\breve{\pi}_\ell$ are
different can be seen for instance when $x=x_0$ is the north pole (so $\vartheta=0$): then $R_{x_0;\ell}$ projects in the direction
of $\ket{\ell,0}$, whereas $Q_{x_0;\ell}$ projects in the direction of the lowest weight vector $\ket{\ell,-\ell}$.

This point has been discussed in some detail in section 4 of \cite{VG89}.
\end{rem}

Since $[UDU^*, P_\ell ]= 0$,  using the quantization map
$\breve{\pi}_\ell$ one applies Lemma \ref{lemma:qAA} and get that the
distance $d_{\A_\ell, D_\ell}$ on the fuzzy sphere is not smaller than the
distance $d_{\A, D}$ on $\mc{S}(C^\infty(\bS^2))$ induced by the pull back
of $\breve{\pi}_\ell$. In the following we will use instead the map
$\pi_\ell$ so that to find -- as in Moyal plane -- a lower bound
$d^\flat_{\A, D}$ to
$d_{\A_\ell, D_\ell}$ and also an upper bound given by the distance
\begin{equation}
d_{\A,D}^{(\ell)}(\varphi,\psi) :=\sup_{f=f^*\in\A}
\big\{ \varphi(f)-\psi(f)\,:\, \|[D,B_\ell(f)]\|\leq 1 \big\} \;
\end{equation}
where $B_\ell:=\sigma_\ell\circ\pi_\ell:\A\to \A$
 is the Berezin transform,
\begin{equation}
B_\ell(f)(x)=\int_{\bS^2}K_x(y)f(y)\de\mu_y=\inner{K_x,f}_{L^2} \;,
\end{equation}
with reproducing kernel $K_x(y)=\gamma_\ell\tr(Q_{x;\ell}Q_{y;\ell})$.
Note that $K_x(y)$ can easily be expressed explicitly as a sum of Legendre polynomials, see for instance Eqn.~(4.7) of Ref.~\cite{VG89}.

\medskip

We begin with technical lemmas, similar to the ones for Moyal plane.
\begin{lemma}
Both the maps $\sigma_\ell$ and $\pi_\ell$ are unital, positive and norm
non-increasing. Moreover $\sigma_\ell:\A_\ell\to\A$ is injective, hence the adjoint map
$\pi_\ell:\A\to\A_\ell$ is surjective and the pull back
$\sharp:\mc{S}(\A_\ell)\to\mc{S}(\A)$ of $\pi_\ell$ on the space of state is injective.
\end{lemma}
\begin{proof}
For $\sigma_\ell$, all properties but injectivity follow from $\Psi_{x;\ell}$ being a state:
$|\sigma_\ell(a)(x)|=|\Psi_{x;\ell}(a)|\leq\|a\|$ for all $x$, hence
$\|\sigma_\ell(a)\|_\infty\leq\|a\|$.
The injectivity of $\sigma_\ell$ is checked by explicit computation.
For $\pi_\ell$, the only non-trivial point is unitality, i.e.~$\pi_\ell(1)=P_\ell$. From \eqref{eq:lemma1}
one has $[\rho_\ell(\xi),\pi_\ell(1)]=0$ for all $\xi\in\mf{su}(2)$. Since the representation $V_\ell$ is irreducible,
by Shur lemma $\pi_\ell(1)=\lambda P_\ell$ is proportional to the identity endomorphism of $V_\ell$. Since $\tr(\pi_\ell(1))=\gamma_\ell\int_{\bS^2}\tr(Q_{x;\ell})\de\mu_x=2\ell+1$
and $\tr(\lambda P_\ell)=\lambda(2\ell+1)$, one gets $\lambda=1$.
\end{proof}
\begin{lemma}\label{lemma:Qeq}
For all $\xi\in\mf{su}(2)$ and $x\in\bS^2$:
\begin{equation}
\partial_\xi Q_{x;\ell}+[\rho_\ell(\xi),Q_{x;\ell}]=0 \;.
\end{equation}
\end{lemma}
\begin{proof}
From \cite[Lem. 4.2]{DLV12},
$\Psi_{x;\ell}([\rho_\ell(\xi),a])=\partial_\xi\Psi_{x;\ell}(a) \; \forall\xi\in\mf{su}(2)$, $a\in\A_\ell$ and $x\in\bS^2$. By  cyclicity of
the trace, 
$\inner{a,\partial_\xi
  Q_{x;\ell}+[\rho_\ell(\xi),Q_{x;\ell}]}_{\mathrm{HS}}=0\; \forall
a\in\A_\ell$, hence the lemma.
\end{proof}
\begin{cor}
For all $\xi\in\mf{su}(2)$, $f\in\A$ and $a\in\A_\ell$ one has {\footnote{When acting on $\HH\otimes\C^2$ and $V_\ell\otimes\C^2$, we
 write $f$ and $a$ for the operators $f\otimes\id_{\C^2}$ and $a\otimes\id_{\C^2}$. Similarly the maps $\pi_\ell\otimes\id_{M_2(\C)}$
 and  $\sigma_\ell\otimes\id_{M_2(\C)}$ are denoted by $\pi_\ell$ and $\sigma_\ell$.}}
\begin{equation}
\pi_\ell(\partial_\xi f)=[\rho_\ell(\xi),\pi_\ell(f)],\quad
\partial_\xi\sigma_\ell(a)=\sigma_\ell([\rho_\ell(\xi),a]) \;.\\
\label{eq:lemma1}
\end{equation}
Hence
\begin{equation}
\label{eq:lemma1bis}
\pi_\ell([D,f])=[D_\ell,\pi_\ell(f)],\quad
[D,\sigma_\ell(a)]=\sigma_\ell([D_\ell,a]).
\end{equation}
\end{cor}
\begin{proof}
The first equation in \eqref{eq:lemma1} comes from Lemma \ref{lemma:Qeq} using integration by parts. The second
comes from Lemma \ref{lemma:Qeq} and the cyclic property of the
trace. Eq.~\eqref{eq:lemma1bis} then follows from the explicit form of $D$
and $D_\ell$.
\end{proof}
Following Prop.~\ref{propBerezin}, one gets the announced upper and
lower bound to the distance on the quantum sphere.
\begin{prop}
\label{ineqfuzzysphere}
For any $\varphi,\psi\in\mc{S}(\A_\ell)$, one has
\begin{equation}\label{eq:prop8}
d_{\A,D}^\flat(\varphi,\psi)\leq d_{\A_\ell,D_\ell}(\varphi,\psi)\leq d_{\A,D}^{(\ell)}(\varphi^{\sharp},\psi^{\sharp}).
\end{equation}
\end{prop}

\begin{rem}
As for Moyal, coherent states converge to points in the weak$^*$ topology.
Indeed for all $f\in\A$ one has
\begin{equation}
\|f-B_\ell(f)\|_\infty\leq L_D(f)\cdot \frac{\pi}{2^{\gamma_\ell+2}}\binom{2\gamma_\ell}{\gamma_\ell} \;,\qquad\forall\;f\in\A\,.
\label{eq:70}
\end{equation}
This is a particular case of \cite [Thm.~2.3]{Rie04b}. The coefficient multiplying
$L_D(f)$ is given by $\int_{\bS^2}d_{\mathrm{geo}}(x_0,y)K_{x_0}(y)\de\mu_y$, and is
independent on $x_0$.
Taking for $x_0$ the north pole $\vartheta=0$,
one has $K_{x_0}(y)=\gamma_\ell(\cos\tfrac{\vartheta'\!}{2})^{4\ell}$,
and $d_{\mathrm{geo}}(x_0,y)=\vartheta'$ where $\vartheta'$ is the polar angle of $y$.
An explicit computation of the integral yields (\ref{eq:70}).
From the asymptotic behaviour $\binom{2n}{n}\sim
\frac{2^{2n}}{\sqrt{\pi n}}$, we deduce that $\Psi_{x;\ell}^\sharp(f)=B_\ell(f)(x)\xrightarrow{\ell\to\infty}f(x)$. \qed
\end{rem}
\medskip

To conclude, we observe that since $\|[D_\ell,\,.\,]\|$ is a Lipschitz seminorm on $\A_\ell$, and the algebra is finite-dimensional,
we know that $(\A_\ell,V_\ell\otimes\C^2,D_\ell)$ is a compact quantum
metric space, so that the distance
$d_{\A_\ell,D_\ell}$ is finite. But $d_{\A,D}^{(\ell)}$ is not: as for
the regularization by eigenprojections of section~\ref{sec:5.3}, 
 the representation $\pi_\ell$ cuts the components of $f$ with high angular
momentum and the distance
between pure states is infinite.
\begin{prop}
\label{prop:fuzzysphere}
  \begin{align}
\label{eq:911}
    d_{\A,D}^{(\ell)}(\delta_x,\delta_y)&=\infty\quad\forall\;x\neq y,\\
    d_{\A,D}^{(\ell)}(\varphi^{\sharp},\psi^{\sharp})&<\infty
    \quad \forall\; \varphi,\psi\in\mc{S}(\A_\ell).
\label{eq:91}
    \end{align}
\end{prop}

\begin{proof}
Let $f_k(x)=e^{-ik\phi}\sin\vartheta$ with $k>2\ell$. Then $\pi_\ell(f)=0$ as one can see
by performing the integral in $\de\phi$ in \eqref{eq:BSpil}.
Let $x=(\phi,\vartheta)$ and $y=(\phi',\vartheta')$.
For $x\neq y$ we can always find a $k>2\ell$ such that $f_k(x)\neq
f_k(y)$, which proves \eqref{eq:911}: 

- when $\vartheta\neq\vartheta'$ or when
$\vartheta=\vartheta'$ and $\frac{\phi-\phi'}{2\pi}$ is irrational, any $k$ is fine;

- when $\vartheta=\vartheta'$ and $\frac{\phi-\phi'}{2\pi}=\frac{p}{q}$ with $p$ and $q$ coprime,
then any $k$ coprime to $q$ is fine.

\noindent
For the same reason above, $\pi_\ell(Y_{km})=0$ if $k>2\ell$ and the support of $\pi_\ell$
is $V=\mathrm{Span}\{Y_{km}:k\leq 2\ell\}$. Now $\dim(V)=(2\ell+1)^2$, and being $\pi_\ell$ surjective,
$\dim\mathrm{Im}(\pi_\ell)=\dim(\A_\ell)=\dim M_{2\ell+1}(\C)=(2\ell+1)^2$ too, proving that the restriction
of $\pi_\ell$ to $V$ is also injective. Since $\sigma_\ell$ is injective, we deduce that the map $B_\ell:V\to V$
is injective too. Therefore $[D,B_\ell(f)]=B_\ell([D,f])$ (from
\eqref{eq:lemma1bis}) is zero only if $[D,f]=0$, i.e.~$f$ is constant.
Hence $L_B := \|[D,B_\ell(\,.\,)]\|$ is a Lipschitz seminorm on
$V$. By construction $\varphi^{\sharp}$ and $\psi^{\sharp}$ depends only on
the component of $f$ belonging to $V$, thus
\begin{equation}
d_{\A,D}^{(\ell)}(\varphi^{\sharp},\psi^{\sharp}) =\sup_{f=f^*\in V^{\mathrm{sa}}}
\big\{ \varphi^{\sharp}(f)-\psi^{\sharp}(f)\,:\, \|[D,B_\ell(f)]\|\leq 1 \big\} \;.
\end{equation}
In fact, since we can add a constant to $f$ without changing $\varphi^{\sharp}(f)-\psi^{\sharp}(f)$
nor $\|[D,B_\ell(f)]\|$, the space $V$ can be replaced by $W=\mathrm{Span}\{Y_{km}:0<k\leq 2\ell\}$ (obtained from $V$
by removing the constant functions multiple of $Y_{00}$). Since $L_B$ and $L_D$ are norms on $W$, and the latter is
finite-dimensional, they are equivalent. In particular, $d_{\A,D}^{(\ell)}(\varphi^{\sharp},\psi^{\sharp})$ is
strongly equivalent to $d_{\A,D}(\varphi^{\sharp},\psi^{\sharp})$, and the latter (the Wasserstein distance) is no greater than $2$
(the diameter of $\bS^2$). This proves \eqref{eq:91}.
\end{proof}

\subsection*{Acknowledgments}\vspace{-4pt}{\small
We thank M.~Bordemann, G.~Dito and M.~Schlichenmaier for discussions and for suggesting some useful references.
P.M.~thanks A.~Roche for constant support.
F.D.\ and F.L.\ are partially supported by the ``Progetto FARO 2010'' of the University of Naples {\sl Federico II}.
F.L.\ acknowledges support by CUR Generalitat de Catalunya under project FPA2010-20807.}

\end{document}